\documentclass[10pt,letterpaper]{article}
\usepackage[top=0.85in,left=1in,footskip=0.75in,marginparwidth=2in]{geometry}

\RequirePackage[OT1]{fontenc}
\RequirePackage{graphicx, color,amsmath, amssymb, amsthm, amsfonts}
\RequirePackage[numbers]{natbib}
\RequirePackage[colorlinks,citecolor=blue,urlcolor=blue]{hyperref}
\RequirePackage{mathpazo}
\RequirePackage{color,bm}
\RequirePackage{dsfont}
\RequirePackage{bigints}
\RequirePackage{enumerate}
\RequirePackage{bbm,lipsum}
\usepackage{subcaption}

\theoremstyle{theorem}
\newtheorem{teo}{Theorem}
\newtheorem{cor}{Corollary}
\newtheorem{prop}{Proposition}
\newtheorem{lem}{Lemma}
\newtheorem*{remark}{Remark}
\theoremstyle{definition}
\newtheorem{midef}{Definition}

\newcommand\numberthis{\addtocounter{equation}{1}\tag{\theequation}}
\newcommand\restr[2]{{
  \left.\kern-\nulldelimiterspace 
  #1 
  \vphantom{\big|} 
  \right|_{#2} 
  }}
\newcommand\eval[3]{{
  \left.\kern-\nulldelimiterspace 
  #1 
  \vphantom{\big|} 
  \right|_{#2}^{#3} 
  }}

\newcommand{\ind}{\operatornamewithlimits{\perp}}

\newcommand{\bb}[1]{\mathbb{#1}}

\newcommand{\e}{\ensuremath{\mathbf{e}}}

\newcommand{\sip}{\bb{P}}

\newcommand{\se}{\ensuremath{\bb{E}}}

\newcommand{\si}{{\ensuremath{\bf{1}}}}

\newcommand{\re}{\ensuremath{\mathbb{R}}}

\newcommand{\prob}[1]{\ensuremath{\sip\! \left[ #1 \right]}}

\newcommand{\probc}[2]{\ensuremath{\sip\! \left[ #1 \, | #2 \right]}}
\newcommand{\esp}[1]{\ensuremath{\se\! \left[ #1 \right]}}
\newcommand{\espc}[2]{\ensuremath{\se\! \left[ #1 | #2 \right]}}

\newcommand{\indi}[1]{\si_{#1}}

\begin{document}

\vspace*{0.35in}

\begin{flushleft}
{\Large
\textbf\newline{Bayesian nonparametric estimation of survival functions with multiple-samples information}
}
\newline
\\
Alan Riva Palacio and 
Fabrizio Leisen
\\
\bigskip
\bf{School of Mathematics, Statistics and Actuarial Sciences, University of Kent}\\
\bf{Sibson Building, Canterbury, Kent CT2 7FS}

\end{flushleft}

\section*{Abstract}
In many real problems, dependence structures more general than exchangeability are required. For instance, in some settings partial exchangeability is a more reasonable assumption. For this reason, vectors of dependent Bayesian nonparametric priors have recently gained popularity. They provide flexible models which are tractable from a computational and theoretical point of view. In this paper, we focus on their use for estimating multivariate survival functions. Our model extends the work of Epifani and Lijoi (2010) to an arbitrary dimension and allows to model the dependence among survival times of different groups of observations. Theoretical results about the posterior behaviour of the underlying dependent vector of completely random measures are provided. The performance of the model is tested on a simulated dataset arising from a distributional Clayton copula. 

\section{Introduction}

Bayesian nonparametric modelling in \textit{survival analysis} problems often relies on the assumption that the times observed are \textit{exchangeable}, see for example \cite{doksum} and \cite{ishwaranjames}. Such assumption fails to hold when we consider events that are pooled from different dependent scenarios. For example, consider patients under the same treatment but in different hospitals. The survival times of patients from the same hospital could be assumed exchangeable. On the other hand, this is not a reasonable assumption when we consider patients from different hospitals since factors specific to each hospital might exert significant influence. 
In general, we can consider that the data is originated from $d$ different but related studies. Formally, we have $d$ sets of observations where the exchangeability assumption is assumed only within each set. In the above cases, it would be more appropriate to assume  a form of dependence called \textit{partial exchangeability} (see Section 2 for a formal account on exchangeability and partial exchangeability). This motivates the extension of Bayesian nonparametric models into a partially exchangeable setting where multiple-samples information could be used. 
\\
\\
Applications of Bayesian nonparametrics in survival analysis go back, for example, to \cite{doksum} and \cite{fergphadia}, who used non-decreasing independent increment processes to construct random survival functions. \cite{dykstra} and \cite{loweng} focused on random hazard rates. More recently, \cite{ishwaranjames} used a general class of random hazard rate-based models, and \cite{nieto} used a general short-term and long-term hazard ratios model.
There is an ongoing effort in Bayesian nonparametrics to propose flexible dependent random probability measures as set forth with the seminal work of \cite{MacEachern}. In survival analysis, for example, \cite{deiorio} introduced a model based on a dependent Dirichlet process. In a partial exchangeable setting, survival analysis models have been used, for example, in \cite{epilijoi} where a dependent two-dimensional extension of the  \textit{neutral to the right} (NTR) model was introduced and in \cite{NL14} where a dependent vector of hazard rates was constructed. \cite{GL2016} introduced a new class of vectors of dependent completely random measures, called \textit{Compound Random Measures}, where the dependence contribution is modelled with a parametric distribution. 
\\
\\
In the seminal work of \cite{doksum}, the NTR model for survival functions was introduced. The NTR model can be expressed in terms of a Completely Random Measure (CRM) $\mu$. This means that when $\mu$ is evaluated at pairwise disjoint sets it gives rise to mutually independent nonnegative random variables. We say that a positive random variable $Y$ has a NTR distribution given by a CRM $\mu$, denoted $Y\sim \text{NTR}(\mu)$, if
$$
S(t)=\probc{Y>t}{\mu}=e^{-\mu(0,t]},
$$
where $\mu$ is such that $$\lim_{t \to \infty} \mu (0,t]=\infty.$$
NTR distributions have several appealing properties, including the independence of normalized increments and posterior conjugacy for censored to the right data. An extension of the NTR model into a partially exchangeable setting was given by \cite{epilijoi} for the $2-$dimensional case. In the present work, we follow the approach of  \cite{epilijoi} and focus on models based on a $d$-dimensional  \textit{ vector of completely random measures} (VCRM). More precisely, we consider  $d$ 
collections of survival times $\{
Y^{(1)}_j\}_{j=1}^\infty , \ldots , \{Y^{(d)}_j\}_{j=1}^\infty$ such that, for $\pmb{t}=(t_1,\ldots , t_d)\in (\re^+)^d$,
\begin{align}\label{model}
&S(\pmb{t})=\probc{Y_{i_1}^{(1)}>t_1,\dots, Y_{i_d}^{(d)}>t_d}{(\mu_1,\dots , \mu_d)}=
e^{-\mu_1(0,t_1]-\dots -\mu_d(0,t_d]} ,
\end{align}
with arbitrary $i_1, \ldots , i_d \in \mathbb{N}\setminus \{0\}$.
This model is convenient for modeling data where the dependence among the entries of the VCRM $\pmb{\mu}=\left( \mu_1, \ldots , \mu_d \right)$ accounts for dependence among the multiple-samples in a partially exchangeable setting.  Furthermore, marginally we recover the NTR model, namely
\begin{equation}\label{NTRuno}
Y^{(i)}_1, \ldots , Y^{(i)}_{n_i}
\stackrel{\text{i.i.d.}}{\sim}
\text{NTR}(\mu_i)
\end{equation}
with $i\in \{1, \ldots , d\}$, $n_i \in \mathbb{N} \setminus \{0\}$. In \eqref{NTRuno} we want to model the dependence of the VCRM $\pmb{\mu}$ in a way that allows us to fix a marginal behavior so to exploit the fact that marginally we recover a NTR model; \textit{L\'evy copulas} are a natural framework to model the dependence structure of VCRM's in such way. 
\\

In this paper we provide a posterior characterization for the above model, see Theorem 1. Similarly to \cite{epilijoi} for 2-dimensional setting, we show that the posterior distribution corresponds to a survival function of the type as in \eqref{model} leading to a conjugacy property. Extensions of some results in \cite{epilijoi} are also provided. We would like to stress that the derivation of such results are not trivial when considering an arbitrary dimension. In particular, Proposition 1 gives a general expression for the Laplace exponent when a L\'evy copula is considered to set the dependence of the VCRM underlying the $d-$dimensional NTR model; Proposition 3 gives an alternative characterization of the multivariate NTR. Furthermore, other theoretical results are proved in order to facilitate the calculation of posterior means when the inferential exercise is implemented. Finally, we illustrate the methodology on a synthetic dataset. 
\\
\\
The paper is organized as follows: Section 2 presents the preliminary notions which are needed in this work. In Section 3 we extend some results in \cite{epilijoi} to the multivariate setting. In particular, we state the posterior characterization of the model and provide some useful corollaries for implementing the posterior inference. In Section 4, an application with synthetic data is illustrated. All the proofs can be found in the appendix. 
\section{Preliminaries}
In this section, we provide some preliminaries about exchangeability, partial exchangeability and vectors of completely random measures which are the building blocks of our Bayesian nonparametric proposal. Furthermore, we will illustrate the concept of a positive L\'{e}vy copula which is useful to model the dependence structure between the components of a vector of completely random measures.

\subsection{Exchangeability and Partial exchangeability}
\noindent Let $\mathbb{Z}$ be a complete and separable metric space, with corresponding Borel $\sigma$-algebra $\mathcal{Z}=\mathcal{B}(\mathbb{Z})$
\begin{midef}
A collection of random variables $\{Z_i \}_{i=1}^\infty$ in $\mathbb{Z}$ is exchangeable if for any permutation $\pi$ of $\{1,\dots,m\}$ we have that
$$
\big\{
Z_1,  \ldots , Z_m
\big\}
\stackrel{\text{d}}{=}
\big\{
Z_{\pi(1)},  \ldots , Z_{\pi(m)}
\big\}.
$$
\end{midef}
\noindent
As highlighted in the Introduction, in several problems the exchangeability assumption appears far too restrictive. In particular, we considered $d$ groups of observations where the order in which they are collected within each group is irrelevant. To describe this setting we resorted to the notion of partial exchangeability, as set forth by \cite{finetti}, that formalizes
the idea of partitioning the entire set of observations into a certain number of classes, say
d, in such a way that exchangeability may be reasonably assumed within each class. For
ease of exposition, we confine ourselves to consider the case where d = 2.
\begin{midef}
The collection of random vectors 
$$
\left\{ \left( Z_i^{(1)}, Z_i^{(2)} \right) \right\}_{i=1}^\infty
$$ in $\mathbb{Z}^2$ is partially exchangeable if, for any $m_1, m_2\geq 1$ and for all  permutations $\pi_1$ and $\pi_2$ of $\{1,\dots,m_1\}$ and $\{1,\dots,m_2\}$ respectively, we have that 
$$
\big\{
Z_1^{(1)},\ldots,Z_{m_1}^{(1)},Z_1^{(2)},\ldots,Z_{m_2}^{(2)}
\big\}
\stackrel{\text{d}}{=}
\big\{
Z_{\pi_1(1)}^{(1)},\ldots,Z_{\pi_1(m_1)}^{(1)},Z_{\pi_2(1)}^{(2)},\ldots,Z_{\pi_2(m_2)}^{(2)}
\big\}.
$$
\end{midef}
\noindent
\subsection{Vectors of completely random measures}
\noindent Given a complete and separable metric space $\mathbb{X}$, with corresponding Borel $\sigma$-algebra $\mathcal{X}=\mathcal{B}(\mathbb{X})$, we call a measure $\mu$ on $(\mathbb{X},\mathcal{X})$ boundedly finite if $\mu (A)<\infty$ for any bounded set $A\in \mathcal{X}$. A random measure is a measurable function from a probability space $(\Omega, \mathcal{F}, \mathbb{P})$ onto $(\mathbb{M}_\mathbb{X},\mathcal{M}_\mathbb{X})$ which is the measure space formed by $\mathbb{M}_{\mathbb{X}}$, the space of boundedly finite measures on $(\mathbb{X},\mathcal{X})$, and its corresponding Borel $\sigma$-algebra $\mathcal{M}_\mathbb{X}$.
In particular we will focus on the class of \textit{completely random measures} as introduced in \cite{kingman}.
\begin{midef}
A random measure $\mu$ on a complete and separable metric space $\mathbb{X}$ with corresponding Borel $\sigma$-algebra $\, \mathcal{X}=\mathcal{B}(\mathbb{X})$ is called a completely random measure (CRM) if for any collection of disjoint sets $\{A_1, \dots, A_n \}\subset \mathcal{X}$ the random variables $\mu(A_1),\dots , \mu(A_n)$ are mutually independent.
\end{midef}
\noindent A CRM $\mu$ has the following representation \citep{kingman}, 
$$\mu =\mu_d +  \mu_r + \mu_{fl},$$
where $\mu_d$ is a deterministic measure, $\mu_{fl}$ is a measure that consists on jumps with possibly random jump heights but fixed jump locations, and 
$$\mu_r=\sum_{i=1}^\infty W_i\delta_{X_i},$$ 
where for $i\in \{1,2,\dots \}$  $X_i  \in \mathbb{X}$ are random  jump locations and $W_i \in \re^+ $ are random jump heights. The measures $\mu_d$, $\mu_{fl}$ and $\mu_r$ are mutually independent. In particular, $\mu_r$ is again a CRM and is characterized by the following Laplace transform
\begin{equation}\label{lap1}
\esp{e^{-\lambda\mu_r(A)}}=e^{-\int_{\re^+\times A}(1-e^{-\lambda s})\nu(\mathrm{d}s,\mathrm{d}x)},
\end{equation}
where $\lambda>0$ and $\nu$ is a measure on $\re^+ \times\mathbb{X}$ such that
\begin{align*}
\int_{\re^+\times A} \min\{s,1\}\nu(\mathrm{d}s,\mathrm{d}x) <\infty,
\end{align*}
for any bounded set $A\in \mathcal{X}$. The measure $\nu$ is usually called the L\'{e}vy intensity of $\mu_r$. In the remainder of this work we only consider CRM's $\mu$ without fixed jump locations nor deterministic part so we take $\mu=\mu_r$ to be solely determined by (\ref{lap1}). In particular we focus on L\'{e}vy intensities $\nu$ which are \textit{homogeneous}, i.e.
$$\nu(\mathrm{d}s,\mathrm{d}x)=\rho(\mathrm{d}s)\alpha(\mathrm{d}x),$$
where $\alpha$ is a non-atomic measure on $\mathbb{X}$ referring to the jump locations and 
$\rho$ is a measure on $\re^+$ referring to the jump heights. A popular example of an homogeneous CRM is the $\sigma$-stable process given by
\begin{align}\label{sigstab}
\nu(\mathrm{d} s, \mathrm{d} x)=\frac{A\sigma s^{-1-\sigma}}{\Gamma(1-\sigma)}\mathrm{d}s\alpha(\mathrm{d}x).
\end{align}
As an illustration, we plot in Figure \ref{fig1} the associated process $\mu(0,t]$ for the $\sigma$-stable process (\ref{sigstab}) with $\alpha(\mathrm{d}x)=\mathrm{d}x$.
\begin{figure}
\centerline{
\includegraphics[width=9cm,height=6cm]{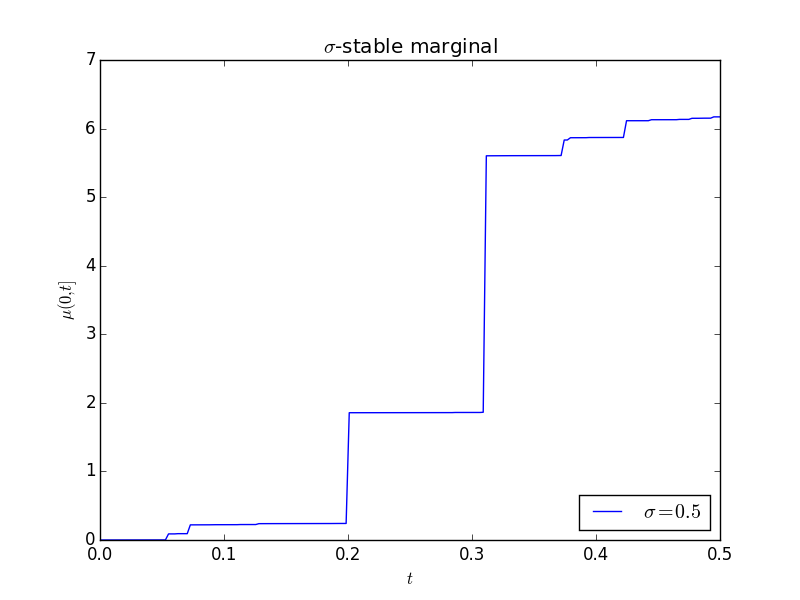}
}
\caption{Plot of $\mu(0,t]$ when a $\sigma$-stable process is considered.}\label{fig1}
\end{figure}
\\

We extend this framework to the multivariate setting by considering vectors $(\mu_1,\dots ,\mu_d)$ where each $\mu_i$ is a homogeneous CRM on $(\mathbb{X},\mathcal{X})$ with respective L\'{e}vy intensities $\bar{\nu}_j(\mathrm{d}s,\mathrm{d}x)=\nu_j(\mathrm{d}s)\alpha(\mathrm{d}x)$. Moreover we take the intensity $\alpha$ to be smooth in the sense that $\alpha((0,t])=\gamma(t)$ with $\gamma:[0,\infty)\rightarrow \re^+$ a  non-decreasing and differentiable function  such that $\gamma(0)=0$ and $\lim_{t\to \infty}\gamma(t)=\infty$; this last conditions on the limit behaviour will enable us to get, marginally, the associated NTR cumulative distributions in our models. We have that for any $A_1,\dots,A_n$ in $\mathcal{X}$, with $A_i\cap A_j=\emptyset$ for any $i\neq j$, the random vectors $(\mu_1(A_i),\dots,\mu_d(A_i))$ and
$(\mu_1(A_j),\dots, \mu_d(A_j))$ are mutually independent; furthermore, one has a multivariate analogue of the Laplace transform (\ref{lap1})
\begin{equation}\label{lap2}
\esp{e^{-\lambda_1\mu_1(A)-\dots -\lambda_d\mu_d(A)}}=e^{-\int_{(\re^+)^d\times A}(1-e^{-\lambda_1 s_1-\dots -\lambda_d s_d})\rho_d(\mathrm{d}s_1,\dots , \mathrm{d}s_d)\alpha(\mathrm{d}x)},
\end{equation}
where $\pmb{\lambda}=(\lambda_1,\dots ,\lambda_d)\in (\re^+)^d$ and $\rho_d$ is a measure on $(\re^+)^d$. In particular, we introduce the notation for the multivariate Laplace transform
\begin{align}\label{laplacenonhomog}
&\esp{e^{-\lambda_1\mu_1(0,t]-\dots -\lambda_d\mu_d(0,t]}}=e^{- \psi_t(\pmb{\lambda}) }.
\end{align}
Henceforth, $\psi_t(\pmb{\lambda})$ is called the Laplace exponent of $\pmb{\mu}=(\mu_1, \dots , \mu_d)$; in the case at hand, 
$\psi_t(\pmb{\lambda})= \gamma(t)\psi(\pmb{\lambda})$ 
where 
$\psi(\pmb{\lambda})=\int_{(\re^+)^d}(1-e^{-<\pmb{\lambda},\pmb{s}>})\rho_d(\mathrm{d}\pmb{s})$
and $<\pmb{\lambda},\pmb{s}>=\sum_{i=1}^d\lambda_is_i$ is the usual inner product in $\re^d$. Marginalizing, we have that
\begin{align*}
&\nu_i(A)=\int_A \nu_i(\mathrm{d}s)=\int_{(\re^+)^{d-1}}\rho_d(\mathrm{d}s_1,\dots,\mathrm{d}s_{i-1},A,\mathrm{d}s_{i+1},\dots, \mathrm{d}s_d).
\end{align*}
In Section 3, we use this particular kind of homogeneous and additive vector of CRM's to construct priors for survival analysis models.

\subsection{Positive L\'{e}vy copulas}
Although in this work we consider vectors of CRM's with fixed marginal behaviour, it remains to establish the dependence structure. \cite{tankov} introduced the concept of positive L\'evy  copulas which allows to construct vectors of CRM's with fixed marginals. 
\begin{midef}
A function $\mathcal{C}(\pmb{s}=(s_1,\dots,s_d)):[0,\infty)^d\rightarrow [0,\infty]$ is a positive L\`{e}vy  copula if
\begin{enumerate}[(i)]
\item $\forall\, B=[s_1,t_1]\times\dots\times[s_d,t_d]\subset [0,\infty)^d $ such that $s_1\leq t_1,\dots,s_d<t_d$ we have that
$$
\sum_{\{\pmb{v}\; :\, \pmb{v} \text{ is a vertex of B}\}} \text{sign}(\pmb{v})\mathcal{C}(\pmb{v})  \geq 0,
$$ 
with
$$
\text{sign}(\pmb{v})=\begin{cases}
                  \;\;\;1,\quad \text{if } v_k = s_k \text{ for an even number of vertices,}
                  \\
                  -1, \quad \text{if } v_k = s_k \text{ for an odd number of vertices.}
                \end{cases}
$$
\item If $\pmb{s}$ is such that $s_i=0$ for some $i\in \{1,\dots , d\}$ then $\mathcal{C}(\pmb{s})=0$.
\item Let $y_1=\dots=y_{k-1}=y_{k+1}=\dots =y_d=\infty$ and 
\\ $C_k(s)=\mathcal{C}(y_1,\dots,y_{k-1},s_k,y_{k+1},\dots,y_d)$ for $k\in\{1,\dots ,d\}$ then $C_k(s)=s$.
\end{enumerate}
\end{midef}
\noindent For example, a vector of independent CRM's is obtained with 
$$
\mathcal{C}_{\ind,d}(\pmb{s})=s_1\indi{s_2=\infty,\dots,s_d=\infty}+\dots +s_d\indi{s_1=\infty,\dots , s_{d-1}=\infty}.
$$
A vector of completely dependent CRM's, in the sense that the jumps of the stochastic vector are in a set $S$ such that whenever $\pmb{v},\pmb{u}\in S$ then either $v_i<u_i$ or $u_i<v_i$ for all $i\in\{1,\dots ,d\}$, is obtained with
$$
\mathcal{C}_{\|,d}(\pmb{s})=\min\{s_1,\dots,s_d\}.
$$
An interesting example of positive L\'{e}vy copulas is the Clayton L\'{e}vy copula 
\begin{equation}\label{clay}
\mathcal{C}_{\theta,d}(\pmb{s})=(s_1^{-\theta}+\cdots +s_d^{-\theta})^{-\frac{1}{\theta}}.
\end{equation}
The parameter $\theta$ is positive and regulates the level of dependence. The above copulas are special cases of the Clayton L\'{e}vy copula, i.e.
$$
\lim_{\theta \to 0}\mathcal{C}_{\theta,d}(\pmb{s})=\mathcal{C}_{\ind,d}(\pmb{s})\text{  and  } \lim_{\theta \to \infty}\mathcal{C}_{\theta,d}(\pmb{s})= \mathcal{C}_{\|,d}(\pmb{s}).
$$
We define the tail integral of an univariate L\'{e}vy intensity $\nu$ to be $U(x)=\int_x^\infty \nu(s)\mathrm{d}s$. In the setting of Section $\pmb{2.1}$ we use a L\'{e}vy copula $\mathcal{C}_d$ and the marginal tail integrals $U_1,\dots,U_d$ associated to $\nu_1,\dots, \nu_d$ to specify an absolutely continuous $\rho_d(\mathrm{d}\pmb{s})=\rho_d(\pmb{s})\mathrm{d}\pmb{s}$ via
\begin{align*}
U(\pmb{x})&=\int_{x_1}^\infty \dots \int_{x_d}^\infty\rho_d(\pmb{s})\mathrm{d}\pmb{s}
\\
& \quad =\int_{x_1}^\infty \dots \int_{x_d}^\infty\restr{\frac{\partial^d}{\partial u_1\cdots \partial u_d}\mathcal{C}_d(\pmb{u})}{u_1=U_1(s_1),\cdots u_d=U_d(s_d)}\nu_1(s_1)\cdots \nu_d(s_d)
\mathrm{d}\pmb{s}.
\end{align*}
Therefore, under suitable regularity conditions, we can recover the multivariate L\'{e}vy intensity from the copula and marginal intensities in the following way
\begin{equation}\label{sklar}
\rho_d(\pmb{s})=\restr{\frac{\partial^d}{\partial u_1\cdots \partial u_d}\mathcal{C}_d(\pmb{u})}{u_1=U_1(s_1),\cdots , x_d=U_d(s_d)}\nu_1(s_1)\cdots \nu_d(s_d).
\end{equation}

\begin{figure}
\centerline{
\includegraphics[width=7.5cm,height=5cm]{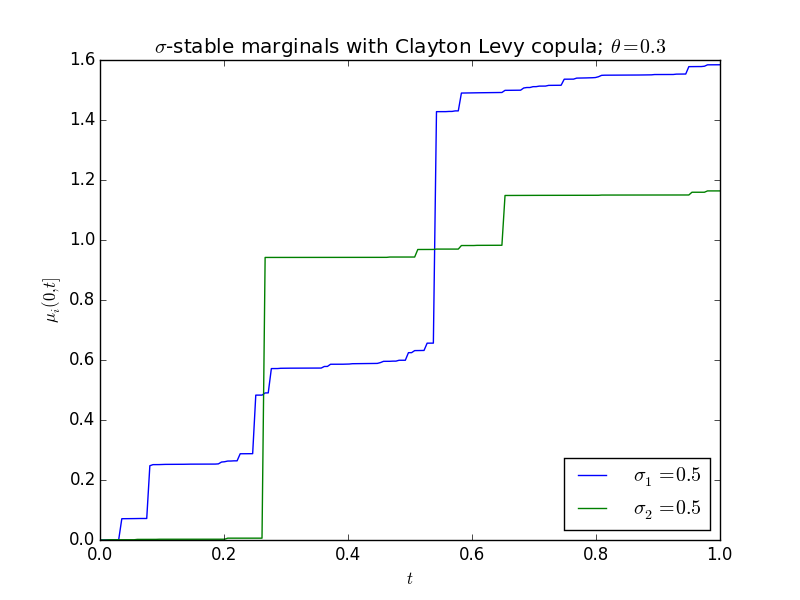}
\includegraphics[width=7.5cm,height=5cm]{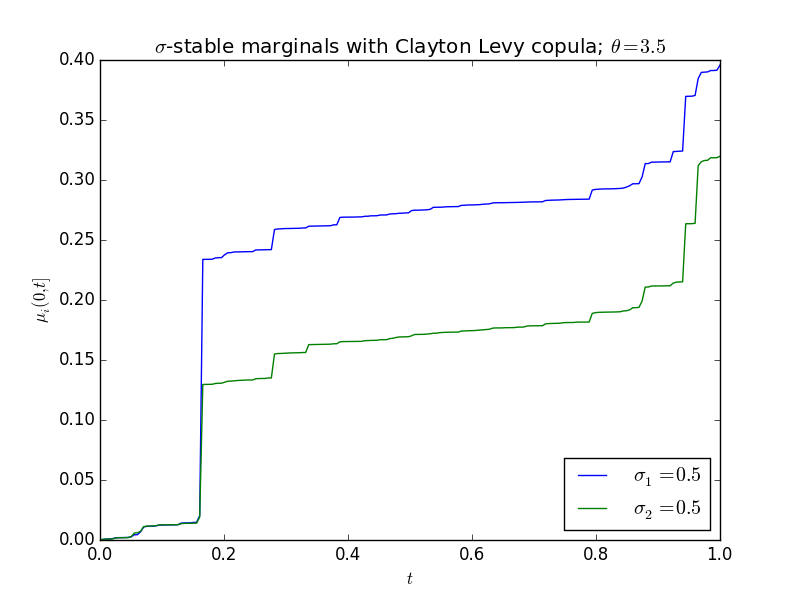}
}
\caption{Plot of dependent $\sigma$-stable processes with dependence given by Clayton L\'{e}vy copula with parameter $\theta=0.3$ (left) and $\theta=3.5$ (right).}\label{fig2}
\end{figure}
For example, consider the Clayton L\'{e}vy copula with $\sigma$-stable margins, given by (\ref{sigstab}), and $\alpha(\mathrm{d}x)=\mathrm{d}x$. Figure \ref{fig2} shows the dependence behaviour when a $2$-dimensional Clayton L\'{e}vy copula with parameter $\theta=0.3$ and $\theta=3.5$ is employed; we plot the associated stochastic processes $\mu_i(0,t]$ with $i\in\{1, 2\}$ similarly to Figure \ref{fig1}. As expected, when $\theta=0.3$, at each jumping time, the processes  have one jump weight big and one small  since we are close to the independence case (where the processes almost surely share no jumping times). On the other hand, when $\theta$ is increased to $3.5$, we can appreciate the higher dependence induced by a larger value of the copula parameter. We simulated the trajectories in Figure 2 by using Algorithm 6.15 in \cite{tankbook}, where a full treatment of the dependence structure of L\'evy intensities is also given. \cite{leisen}, \cite{LLS} and \cite{weixuan} used a L\'evy copula approach for building vectors of dependent completely random measures. 
 
\subsubsection{Working example}
If we consider the L\'{e}vy intensity arising from (\ref{sklar}) when considering the $d$-dimensional Clayton L\'{e}vy copula, (\ref{clay}), with parameter $\theta$ and $\sigma$-stable marginals, (\ref{sigstab}), with parameters $A,\, \sigma$, we obtain
\begin{align*}
\rho_{d, \theta, A, \sigma}(\pmb{s})=
\frac{A(1+\theta)(1+2\theta)\cdots (1+(d-1)\theta)\sigma^d\left( s_1 s_2 \cdots s_d \right)^{\sigma \theta -1}}{\Gamma(1-\sigma)\left( s_1^{\sigma \theta} + \dots +s_d^{\sigma \theta} \right)^{\frac{1}{\theta} +d} }.
\end{align*}
Furthermore, if we take $\theta = 1/\sigma$ we obtain the simplified L\'{e}vy intensity
\begin{align}\label{intens}
\rho_{d, A, \sigma}(\pmb{s})=
\frac{A(\sigma+1)(\sigma+2)\cdots (\sigma+d-1)\sigma}{\Gamma(1-\sigma)\left( s_1 + \dots +s_d \right)^{\sigma +d} }.
\end{align}
Such intensity corresponds to a particular family of vectors of completely random measures known as \textit{Compound Random Measures} (CoRM's) and introduced in \cite{GL2016}; the previous 
L\'{e}vy intensity arises when taking $\phi=1$ in equation (4.4) of the aforementioned paper. A convenient feature of this L\'{e}vy intensity is that, as shown in Proposition 3.1 of \cite{weixuan}, we can explicitly get the corresponding Laplace exponent
\begin{align}\label{lapexp_workex}
\psi_{d, A, \sigma}(\pmb{\lambda})=\sum_{i=1}^d \frac{\lambda_i^{\sigma+d-1}}{\prod_{j=1, \, j\neq i}^d (\lambda_i-\lambda_j)}; \quad \lambda_i \neq \lambda_j \text{ for } j\neq i,
\end{align}
where we take the appropriate limits when $\pmb{\lambda}=(\lambda_1,\dots , \lambda_d)$ is such that $\lambda_i =\lambda_j$ for distinct $i,j\in\{1,\dots,d\}$. As indicated in the remark at the end of section 3, evaluation of the Laplace exponent is necessary for the explicit calculation of the posterior mean of the survival function given censored data.
\section{Main results}
Let $d\in \mathbb{N}\setminus \{0\}$, and suppose we have $d$ collections of random variables 
\begin{equation}
 \{ \{Y^{(i)}_j\}_{j=1}^\infty \}_{i=1}^d.
\end{equation} 
We characterize the probability distribution of these random variables in terms of a vector of CRM's $\pmb{\mu}=(\mu_1,\dots, \mu_d)$. For $\pmb{t}=(t_1,\dots , t_d)\in (\re^+)^d$, let
\begin{align*}
&\probc{Y_1^{(1)}>t_{1,1}, \ldots , Y_{n_1}^{(1)}>t_{1,n_1}, \ldots,
Y_1^{(d)}>t_{d,1}, \ldots , Y_{n_d}^{(d)}>t_{d,n_d} }{(\mu_1, \ldots , \mu_d)}
\\
&=\prod_{i=1}^d\prod_{j=1}^{n_i}e^{-\mu_i (0,t_{i,j}]}. \numberthis \label{model1}
\end{align*}
We observe that under such model the random variables (\ref{model1}) are partially exchangeable and marginally follow a $NTR$ process. The dependence structure in this model can be given through the L\'{e}vy copula associated to the CRM $\pmb{\mu}$. This model extends the one in \cite{epilijoi} to an arbitrary dimension $d$.

The family of Clayton L\'{e}vy copulas is of interest because it has both the independence and complete dependence cases as limit behaviour. In the next result, we work towards finding expressions for the Laplace exponent associated to the Clayton family in such a way that the dependence structure is decoupled across dimensions. This result could be useful since, as we will see, an explicit calculation of $\psi$ is of key importance to implement the Bayesian inference in our survival analysis model.\\

Let $\rho_d(\pmb{s};\theta)$ be the L\'{e}vy intensity associated  via (\ref{sklar}) to the  Clayton L\'{e}vy copula
$\mathcal{C}_{\theta,d}$ and fixed marginal L\'{e}vy intensities $\nu_1,\dots , \nu_d$ with corresponding Laplace transforms $\psi_1, \dots , \psi_d$. We denote the vector of tail integrals corresponding to the marginal L\'{e}vy intensities as $\pmb{U}_d(\pmb{x})=(U_1(x_1),\dots,U_d(x_d))$ and fix the notation
\begin{align*}
\kappa(\theta;\pmb{\lambda}, \pmb{i})&=\lambda_{i_1}\cdots\lambda_{i_m} \int_{(\re^+)^m} \e^{-\lambda_{i_1}s_1 - \dots -\lambda_{i_m}s_m  }C_{\theta,m}(U_{i_1}(s_1), \dots , U_{i_m}(s_m))\mathrm{d}\pmb{s} ,
\end{align*}
where $d\in \mathbb{N}\setminus \{0\}$, $\pmb{\lambda}=(\lambda_1,\dots , \lambda_d)\in (\re^+)^d$, $m\in \{1,\dots , d\}$,  and $\pmb{i}=(i_1,i_2,\dots,i_m)\in \{ 1, \dots , d \}^m$ is such that $i_1< \dots< i_m$.
\begin{prop}
Suppose that $d \in \{2,3,\dots\}$ and
\begin{equation*}
\int_{\|\pmb{s} \|\leq 1}\|\pmb{s}\| \rho_d(\pmb{s};\theta)\mathrm{d}\pmb{s} <\infty , \numberthis \label{finit}
\end{equation*}
then 
\begin{align*}
\psi(\pmb{\lambda})&=\int_{(\re^+)^d}(1-\e^{-<\pmb{\lambda},\pmb{s}>})
\frac{\partial^d}{\partial u_d\cdots\partial u_1}\restr{C_{\theta,d}(\pmb{u})}{\pmb{u}=\pmb{U}_d(\pmb{s})}\nu_1(s_1)\cdots \nu_d(s_d)\mathrm{d}\pmb{s}
\\
&=
\sum_{i=1}^d\psi_i(\lambda_i)-
\sum_{\stackrel{\pmb{i}=(i_1,i_2)\in\{1,\dots,d\}^2}{i_1< i_2}}\kappa(\theta;\pmb{\lambda}, \pmb{i})+\cdots
\\
&\qquad\cdots +(-1)^d
\sum_{\stackrel{\pmb{i}=(i_1,\dots ,i_{d-1})\in \{1,\dots,d\}^{d-1}}{i_1<\dots <i_{d-1}}} \kappa(\theta;\pmb{\lambda}, \pmb{i}) 
+(-1)^{d+1}\kappa(\theta;\pmb{\lambda}, (1,\dots, d)),
\end{align*}
where $\pmb{\lambda}=(\lambda_1,\dots , \lambda_d)\in (\re^+)^d$.
\end{prop}
\noindent We refer to the Appendix \ref{proof1} for the proof. We  incorporate the L\'{e}vy exponent $\psi$ in the multivariate survival analysis setting of (\ref{model1}),
in the next result. We introduce the notation
\begin{equation*}
\nu_{i_1 ,\dots ,i_h}
(s_{i_1}, \dots , s_{i_h})
=\int_0^\infty \cdots \int_0^\infty \rho_d(\pmb{s})\prod_{j\not \in \{i_1,\dots ,i_h\} } \mathrm{d}s_{j}
\end{equation*}
for $h\in \{1,\dots , d\}$ and distinct $i_1,\dots,i_h\in\{1,\dots ,d\}$; and denote $\psi_{i_1, \cdots ,i_h}$ for the respective Laplace exponents.
\begin{prop}
In the context of (\ref{model1}), let $\pmb{1}=(1,\dots, 1)$. For $t_1\leq \cdots \leq t_d$ and $i_1,\dots,i_d\in\{1,\dots ,d\}$ such that $t_{i_1}\leq \dots \leq t_{i_d}$ then
\begin{align}
\prob{Y^{(1)}>t_1,\dots, Y^{(d)}>t_d}&=\nonumber\\\e^{-\gamma(t_{i_1})\psi(\pmb{1})}&\e^{-[\gamma(t_{i_2})-\gamma(t_{i_1})]\psi_{i_2,\dots, i_d}(\pmb{1})}\cdots \e^{-[\gamma(t_{i_d})-\gamma(t_{i_{d-1}})]\psi_{i_d}(\pmb{1})}.
\end{align}
\end{prop}
\noindent 
We refer to the Appendix \ref{proof2} for the proof. This result showcases the importance of the Laplace exponent $\psi$ for calculating probabilities in the model and the impact of the function $\gamma(t)$, related to the time depending part of the  Laplace exponent, in the survival function. In Section 4, we will show that the availability of the Laplace exponent is also of main importance to implement the Bayesian inference for the model. The model we are working on generalizes to arbitrary dimension the classic model of \cite{doksum}. We present a multivariate extension of Theorem 3.1 in \cite{doksum}, which relates our model with the notion of neutrality to the right. Let $F$ be a $d$-variate random distribution function on $(\re^+)^d$ and, for a $d$-variate vector of CRM's $\pmb{\mu}=(\mu_1,\dots , \mu_d)$, denote $\mu_{i}(t)=\mu_i\left((0,t] \right)$ with $i\in\{1,\dots ,d\}$. Then, we have the following multivariate extension to Theorem 3.1 in \cite{doksum} and Proposition 4 in \cite{epilijoi}.
\begin{prop}
$F(\pmb{t}=(t_1,\dots,t_d))$ has the same distribution as $$[1-\e^{-\mu_1(t_1)}]\cdots[1-\e^{-\mu_d(t_d)}]$$
for some $d$-variate CRM $\pmb{\mu}=(\mu_1,\dots , \mu_d)$ if and only if for $h\in\{1,2,\dots\}$ and vectors $\pmb{t}_1=(t_{1,1},\dots, t_{d,1}),\dots,$ $\pmb{t}_h=(t_{1,h},\dots,
t_{d,h}) $ with 
$t_{0,i}=0<t_{1,i}<\cdots <t_{d,i}$ and 
$t_{j,0}=0<t_{j,1}<\cdots <t_{j,h}$, there exists $h$ independent random vectors $(V_{1,1},\dots V_{d,1}),\dots , (V_{1,h},\dots V_{d,h})$ such that
\begin{align*}
&(F(\pmb{t}_1),\dots ,F(\pmb{t}_h))\stackrel{d}{=}
\\&
\;
\left(
V_{1,1}\cdots V_{d,1},[1-\bar{V}_{1,1} \bar{V}_{1,2}]\cdots [1-\bar{V}_{d,1}\bar{V}_{d,2}]
,\dots , [1-\prod_{j=1}^h\bar{V}_{1,j}]\cdots [1-\prod_{j=1}^h\bar{V}_{d,j}]
\right), \numberthis \label{ntr}
\end{align*}
where $\bar{V}_{i,j}=1-V_{i,j}$ with $i\in \{1,\dots,d\}$ and $j\in \{1,\dots ,h\}$.
\end{prop}
\noindent We refer to the Appendix \ref{proof3} for the proof. We now establish some notation in order to address the posterior distribution arising from (\ref{model1}) when some survival data is available. Let $\bm{Y}_{n_i}^{(i)}=\left(Y^{(i)}_1,\dots,Y^{(i)}_{n_i}\right)$, $i=1,\dots,d$, be $d$ groups of observations that come from the distribution given by
\begin{align*}
\probc{\bm{Y}_{n_1}^{(1)}>\bm{t}_{1,n_1}, \dots, \bm{Y}_{n_d}^{(d)}>\bm{t}_{d,n_d}}{(\mu_1,\dots ,\mu_d)}= \prod_{i=1}^d\prod_{j=1}^{n_i}
\e^{-\mu_i(0,t_{i,j}]},
\end{align*}
where $\bm{t}_{i,n_i}=\left( t_{i,1},\dots,t_{i,n_i}\right)$ and the event $\{\bm{Y}_{n_i}^{(i)}>\bm{t}_{i,n_i}\}$ corresponds to the event $\{Y^{(i)}_1>t_{i,1},\dots,Y^{(i)}_{n_i}>t_{i,n_i}\}$. Let $c^{(1)}_1,\dots, c^{(1)}_{n_1},\dots,c^{(d)}_1,\dots ,c^{(d)}_{n_d} $ be their respective censoring times; therefore, the set of  censored data is the following
$$\pmb{D}=\bigcup_{i=1}^d\{(T^{(i)}_j,\delta^{(i)}_j)\}_{j=1}^{n_i},$$
where $T^{(i)}_j=\min\{Y^{(i)}_j,c^{(i)}_j\}$ and  $\delta^{(i)}_j=\mathbbm{1}_{(0,c^{(i)}_j]}\left(Y^{(i)}_j \right)$. The number of exact observations is $n_e=\sum_{i=1}^d\sum_{j=1}^{n_i}\delta^{(i)}_j$ and the number of censored observations is $n_c=n_1+n_2-n_e$. Taking into account the possible repetition of values among the observations, we consider the order statistics $(T_{(1)},\dots,T_{(k)})$ of the distinct observations where $k$ is the number of distinct observed times among all groups.
\\

\noindent Let define the set functions
\begin{align*}
m_i^e(A)=\sum_{j=1}^{n_i}\delta_j^{(i)}\mathbbm{1}_A(T_j^{(i)})\quad &;\quad 
m_i^c(A)=\sum_{j=1}^{n_i}(1-\delta_j^{(i)})\mathbbm{1}_A(T_j^{(i)})
\end{align*}
for $i\in \{1,\dots,d\}$, which denote the number of, respectively, exact and censored marginal observations in $A$, with respect to group $i$. We define $N_i^e(x)=m_i^e\left( (x,\infty)\right)$,  $\,N_i^c(x)=m_i^c\left( (x,\infty)\right)$, for $i\in \{1,\dots , d\}$ and $\,n_{i,j}^e=m_i^e(\{T_{(j)}\})$, $\,n_{i,j}^c=m_i^c(\{T_{(j)}\})$, $\,\bar{n}_{i,j}^e=\sum_{r=j}^k n_{i,r}^e$ $\,\bar{n}_{i,j}^c=\sum_{r=j}^k n_{i,r}^c$  for $(i,j)\in\{1,\dots,d\}\times\{1,\dots,k\}$; and the corresponding vectors $\bar{\pmb{n}}^e_j=(\bar{n}^e_{1,j},\dots,\bar{n}^e_{d,j})$,
$\bar{\pmb{n}}^c_j=(\bar{n}^c_{1,j},\dots,\bar{n}^c_{d,j})$, for $j\in \{1,\dots,k\}$ and $\pmb{N}^e(x)=(N_1^e(x),\dots, N_d^e(x))$, $\pmb{N}^c(x)=(N_1^c(x),\dots, N_d^c(x))$.
\\
The next theorem determines the calculation of the posterior distribution for a vector of CRM's given some censored data and it applies to general vectors of CRM's. In particular, the assumption that the respective L\`{e}vy  intensity is homogeneous has been dropped.
\begin{teo}
Let $\pmb{\mu}=(\mu_1,\dots , \mu_d)$ be a $d$-variate CRM such that its corresponding 
L\`{e}vy  intensity $\nu(\pmb{s},\mathrm{d}t)\mathrm{d}\pmb{s}$ is differentiable with respect to $t_0$ on $\re^+ \setminus \{0\}$ in the sense that for $\eta_t=\nu(\pmb{s},(0,t])$ the partial derivative
$\eta'_{t_0}(\pmb{s})=\partial\restr{\eta_t(\pmb{s})/\partial t}{t=t_0}$ exists. Moreover we assume that the entries of $\pmb{\mu}$ are not independent.  Then the posterior distribution of $\pmb{\mu}$ given data $\pmb{D}$ is the distribution of the random measure
\begin{equation*}
(\mu_1^\star,\dots ,\mu_d^\star)+\sum_{\{j\,:\,T_{(j)} \text{is an exact observation}\}}
(J_{1,j}\delta_{T_{(j)}},\dots , J_{d,j}\delta_{T_{(j)}} )
\end{equation*}
where
\begin{enumerate}[i)]
\item $\pmb{\mu}^\star=(\mu_1^\star,\dots , \mu_d^\star)$ is a $d$-variate CRM with L\'evy  intensity $\nu^\star$ such that
\begin{equation*}
\restr{ \nu^\star (\mathrm{d}\pmb{s}, \mathrm{d}x) }{ \mathrm{d}x \in (T_{(j-1)}, T_{(j)} ) }
= 
\e^{-\langle \bar{\pmb{n}}_{j}^c+\bar{\pmb{n}}_{j}^e\, , \, \pmb{s} \rangle }
\nu(\mathrm{d}\pmb{s}, \mathrm{d}x)
\end{equation*}
for $j\in \{1, \dots , k+1\}$.
\item The vectors of jumps $\{(J_{1,j},\dots ,J_{d,j})\}_{j\in J}$, with $J=\{j\,:\,T_{(j)} \text{ is an exact observation}\}$, are mutually independent and the vector of jumps corresponding to the exact observation $T_{(j)}$ has density
\begin{equation*}
f_j(\pmb{s})\propto  \prod_{i=1}^d \left\{ \e^{-(\bar{n}_{i,j}^c+\bar{n}_{i,j+1}^e)s_i}(1-\e^{-s_i})^{n_{i,j}^e}\right\}
 \eta'_{T_{(j)}}(\pmb{s}).
\end{equation*}
\item The random measure $\pmb{\mu}^\star$ is independent of $\{(J_{1,j},\dots ,J_{d,j})\}_{j\in J}$, with 
\\
$J=\{j\,:\,T_{(j)} \text{ is an exact observation}\}$.
\end{enumerate}
\end{teo}
\noindent We refer to the Appendix \ref{proof4} for the proof. The previous result showcases that the posterior distribution arising from (\ref{model1}) can be modeled in the same framework via a vector of CRM's by updating the prior vector of CRM's $\pmb{\mu}$ to $\pmb{\mu}^\star$ as above.\\

This result is enough to provide a scheme for posterior inference. In particular, in the setting of (\ref{model1}) and Theorem 1, we want to estimate the corresponding survival function $\probc{Y^{(1)}>t_1,\dots , Y^{(d)}>t_d}{(\mu_1,\dots, \mu_d)}$ when multiple samples information is available. 

A natural approach in Bayesian nonparametrics is to marginalize over the infinite dimensional random element which characterizes the probability model. In our case, given censored data $\pmb{D}$, we calculate the mean of the survival function given the data by marginalizing over the vector of CRM's $\pmb{\mu}$. As a result of Theorem 1, we can calculate such quantity. The next results allow us to implement the necessary inferential scheme for performing the estimation of the survival function as a posterior mean. 
We denote $\pmb{e}_i$ for the canonical basis of $\re^d$, and 
$S_L(t)=S(t\sum_{l\in L}\pmb{e}_l)$ for $t>0$, $\emptyset \neq L \subset \{1,\dots, d\}$.
In view of the independent increments of the CRM's, calculation of the posterior mean of $S_L$ is all that is needed for the evaluation of the posterior mean of $S$. The next corollary shows how to evaluate the posterior mean of $S_L$.
\begin{cor}
Let $\pmb{\mu}$ be a vector of CRM's with corresponding 
L\`{e}vy  intensity such that $\eta_t(\pmb{s})=\gamma(t)\nu(\pmb{s})$ with $\gamma$ a differentiable function satisfying $\gamma '(t)\neq 0$ for $t>0$. Moreover we assume that the entries of $\pmb{\mu}$ are not independent.
Let  $\emptyset \neq L \subset \{1,\dots, d\}$  and set 
$$J_t=\{j\,:\,T_{(j)}\leq t\}$$
where $T_{(k+1)}=\infty$. Then,
\begin{align*}
\hat{S}_L(t)&=\espc{\espc{S_L(t)}{\pmb{\mu}}}{\pmb{D}}=
\e^{- \sum_{j=1}^{k+1}
\left[
\gamma(t \wedge T_{(j)})-\gamma(T_{j-1})
\right] \indi{[T_{(j-1)},\infty )}(t)
\psi_j^\star \left(\sum_{l\in L}\pmb{e}_l \right)
}
\\
 &
\times 
\prod_{j\in J_t}
\gamma '(T_{(j)}) 
\left[
\frac{
 \int_{(\re^+)^d} \prod_{i=1}^d \left\{ \e^{-(\indi{i \in L}+\bar{n}_{i,j}^c+\bar{n}_{i,j+1}^e)s_i}(1-\e^{-s_i})^{n_{i,j}^e}\right\}
 \nu(\pmb{s})\mathrm{d}\pmb{s}
}
{
 \int_{(\re^+)^d} \prod_{i=1}^d \left\{ \e^{-[\bar{n}_{i,j}^c+\bar{n}_{i,j+1}^e]s_i}(1-\e^{-s_i})^{n_{i,j}^e}\right\}\nu(\pmb{s}) \mathrm{d}\pmb{s}
}
\right]
\end{align*} 
where $T_{(0)}=0$ and for $\pmb{\lambda}\in (\re^+)^d$
\begin{align*}
\psi_j^\star(\pmb{\lambda})&=\int_{(\re^+)^d}\left(
1-\e^{-\langle \pmb{\lambda}, \pmb{s} \rangle}
\right)
 \e^{-\langle \bar{\pmb{n}}_j^c +  \bar{\pmb{n}}_j^e,\pmb{s} \rangle}
\nu(\pmb{s})\mathrm{d}\pmb{s}
\\
&=\psi(\pmb{\lambda} + \bar{\pmb{n}}_j^c +  \bar{\pmb{n}}_j^e)-\psi( \bar{\pmb{n}}_j^c +  \bar{\pmb{n}}_j^e).
\end{align*}
\end{cor}
We see that we can estimate $S(\pmb{t})$ for arbitrary $\pmb{t}\in (\re^+)^d$ in terms of the estimates defined in the previous corollary. Indeed, let $\pmb{t}=(t_1,\dots , t_d)$ and $\pi$ be a permutation of $\{1,\dots ,d\}$ such that $t_{\pi(1)}\leq t_{\pi(2)}\leq \dots \leq t_{\pi(d)}$. We define, for $i\in \{1,\dots, d-1\}$,  the following sets 
\begin{align*}
& L_i = \{j\, : \,  \pi^{(-1)}(j)\geq i \}.
\end{align*}
From the independence of increments of CRM's, it follows that the 
posterior mean of the survival function given censored data $\pmb{D}$ is
\begin{align}\label{estim}
\hat{S}(\pmb{t})=\espc{\espc{S(\pmb{t})}{\pmb{\mu}}}{\pmb{D}}=
\hat{S}_{L_1}(t_{\pi (1)})\prod_{i=1}^{d-1}
\frac{\hat{S}_{L_{i}}(t_{\pi(i+1)})}
{\hat{S}_{L_{i}}(t_{\pi(i)})}\quad \pmb{t}\in (\re^+)^d.
\end{align}
Usually, we deal with  L\'{e}vy intensities which exhibit some dependences in a vector of hyper-parameters $\pmb{c}$. On the proof of Theorem 1, it is outlined how, given censored data $\pmb{D}$ as before, we could derive the likelihood of the hyper-parameters in the L\'{e}vy intensity. This likelihood is necessary for implementing the inferential procedure and it is displayed in the next corollary. 
\begin{cor}
Let $\pmb{\mu}$ be a vector of CRM's with corresponding 
L\`{e}vy  intensity such that $\eta_t(\pmb{s})=\gamma(t)\rho_{d, \pmb{c}}(\pmb{s})$ with $\gamma$ a differentiable function satisfying $\gamma '(t)\neq 0$ for $t>0$, and $\pmb{c}$ a vector of hyper-parameters. Given censored data $\pmb{D}$ we get the likelihood on $\pmb{c}$.
\begin{align*}
l(\pmb{c};\mathcal{D})=&\e^{-\sum_{j=1}^k
\left[ \gamma(T_{(j)})-\gamma(T_{(j-1)}) \right]\psi_{\pmb{d,c}}( \bar{\pmb{n}}_j^c + \bar{\pmb{n}}_j^e )}
\\
\quad \times &
\prod_{j\in J}\gamma ' (T_{(j)})\int_{(\re^+)^d}\prod_{i=1}^d
\left\{
\e^{-(\bar{n}_{i,j}^c + \bar{n}_{i,j}^e)s_i}(1-\e^{-s_i})^{n_{i,j}^e}
\rho_{d,\pmb{c}}(\pmb{s})\mathrm{d}\pmb{s}
\right\},
\end{align*} 
where $\psi_{d,\pmb{c}}$ is the Laplace exponent associated to $\rho_{d,\pmb{c}}$.
\end{cor}
The next lemma provides a useful identity for the computation of the integrals in Corollary 1 and Corollary 2.
\begin{lem}
For $\pmb{q}=(q_1,\dots , q_d)\in (\re^+)^d$ and $\pmb{n}=(n_1,\dots, n_d)\in \mathbb{N}^d$
\begin{align*}
&\int_{(\re^+)^d} \e^{-\langle \pmb{q}, \pmb{x} \rangle }
\prod_{i=1}^d\left(
1-\e^{-s_i}
\right)^{n_i}\nu(\pmb{s})\mathrm{d}\pmb{s}
= 
\sum_{i=1}^d \sum_{k=1}^{n_i}\binom{n_i}{k}(-1)^{k+1}
[\psi(k\pmb{e}_i + \pmb{q})-\psi(\pmb{q})]
\\
&+\sum_{
\begin{array}{c}
i_1\neq i_2 \\ 
n_{i_1},n_{i_2} \notin \{0\} 
\end{array} 
}
\sum_{k_1=1}^{n_1}\sum_{k_2=1}^{n_2}\binom{n_1}{k_1} \binom{n_2}{k_2}
(-1)^{k_1+k_2+1}[\psi(k_1\pmb{e}_{i_1}+k_2\pmb{e}_{i_2}+\pmb{q}) - \psi(\pmb{q})]
\\
&+\dots 
\\
& +
\indi{\{n_1\neq 0, \dots, n_d\neq 0\}}
\sum_{k_1=1}^{n_1}\dots \sum_{k_d=1}^{n_d}(-1)^{k_1+\dots +k_d+1}
[\psi(k_1\pmb{e}_1+\dots +k_d\pmb{e}_d+\pmb{q})-\psi(\pmb{q})].
\end{align*}
\end{lem}
We omit the proof as it is just an application of the binomial theorem in the same line as the proof of Lemma 3 in the appendix.
\begin{remark} The previous results highlights that the implementation of the inferential procedure depends on whether we can perform evaluations of the Laplace exponent or not.
\end{remark}
\section{Applications}
In this section we perform the fitting of a multivariate survival function given censored to the right data in the framework of (\ref{model1}). As mentioned in the previous remark, the evaluation of the Laplace exponent of $\pmb{\mu}$ in (\ref{model1}) is necessary to evaluate the posterior mean in Corollary 1 and the likelihood in Corollary 2; with this in mind, we choose the random measure $\pmb{\mu}$ given by the L\'evy intensity showcased in (\ref{intens}), so that the corresponding Laplace exponent is readily given by (\ref{lapexp_workex}).  For illustration purposes, we use 4-dimensional data arising from a distributional copula with fixed marginal distributions, see \cite{nelsen} for an overview of distributional copulas. More precisely, we generate simulated data $\pmb{Y}=(Y_1,...,Y_4)$ with probability distribution $F_{\theta,\lambda}$ given by a distributional Clayton copula with parameter 
$\theta$ and exponential marginals with parameter $\lambda$. Then, we perform right-censoring by considering censoring time variables $\pmb{c}$ consisting of independent exponential random variables with parameter $\lambda_c$, and define
\begin{align*}
\pmb{\delta}&=(\indi{Y_1<c_1}, \dots, \indi{Y_4<c_4}),
\\
\pmb{T}&=(\min\{ Y_1, c_1 \}, \dots ,
\min\{Y_4, c_4  \} ). \numberthis \label{survdat}
\end{align*}
\\
For fitting the data, we use the 4-dimensional L\'evy intensity given by (\ref{intens}) and  assign priors for the hyper-parameters in (\ref{intens}), $\sigma$ and $A$. We choose a log-normal prior for the parameter A and a Beta prior for the parameter $\sigma$. We use the Metropolis within Gibbs algorithm to draw samples from the posterior distributions of $A$ and $\sigma$ by making use of the likelihood showed in Corollary 2. We present a Monte Carlo approximation of the estimator (\ref{estim}), where we have averaged over the posterior draws of $A$ and $\sigma$. A more in depth description of the simulation algorithm is given in Appendix \ref{algorithmappendix}. In Figures 3 and 4 we show the fit for 150 possibly right censored observations as in (\ref{survdat}). The simulated synthetic observations are such that
\begin{align*}
\pmb{Y}_j & \sim F_{\theta = 0.3, \lambda =  1.}, & j=1, \dots , 150
\\
c_{i,j} & \sim \text{Exp}(\lambda_c = 3.7), & i=1,...,4; \quad j=1, \dots , 150
\\
T_{i,j}  & = \min\{ Y_{i,j}, c_{i,j} \}, & i=1,...,4; \quad j=1, \dots , 150.
\end{align*}
We chose $\lambda_c=3.7$ so we have at least $75\%$ of exact observations for $\pmb{T}$ in each dimension. 
The construction of $F_{\theta,\lambda}$ through a distributional Clayton allows us to calculate explicitly the associated survival function as showcased in Appendix \ref{survappendix}. We use the true survival function for comparison with the fitted survival functions. The estimated survival function are given by the posterior mean 
$$\hat{S}(t_1,t_2,t_3,t_4)= \espc{\espc{S(t_1,t_2,t_3,t_4)}{\pmb{\mu}}}{\pmb{D}},$$
as in \eqref{estim}. The prior distributions of the hyperparameters are
\begin{align*}
\sigma & \sim \text{ Beta}(\mu = 0.4, \sigma^2 = 0.1)
\\
A & \sim \text{ Log-Norm}(\mu =\log (0.88), \sigma^2=3.5).
\end{align*}
We ran $1000$ iterations for the associated Metropolis within Gibbs sampler. Figure 3 and Figure 4 show that the estimated survival functions approximate well the true functions. For comparison purposes, we presented a Kaplan-Meier estimator for the true survival function, see for example \cite{aalen}. As there is no multivariate Kaplan-Meier, we use the next estimator for a multivariate survival function:
\begin{align*}
& \hat{S}_{\text{KM}}(t_1, \ldots , t_d)= 
\\
& S_{\text{KM}}(t_1|T_2>t_2, \ldots, T_d>t_d)
S_{\text{KM}}(t_2|T_3>t_3, \cdots, T_d>t_d)\ldots S_{\text{KM}}(t_d),
\end{align*}
where each $S_{\text{KM}}$ estimator is treated as a univariate Kaplan-Meier estimator restricted to the corresponding set of observations. In Figure 3 and Figure 4, we could appreciate that in the last subplots of each column the Kaplan-Meier can fit poorly as there are less observations on the conditioned Kaplan-Meier functions, as presented in the formula above. 
\begin{figure}
\centering
\begin{subfigure}{0.48\textwidth}
\centering
	 \raisebox{-\height}{\includegraphics[width=0.9\textwidth]{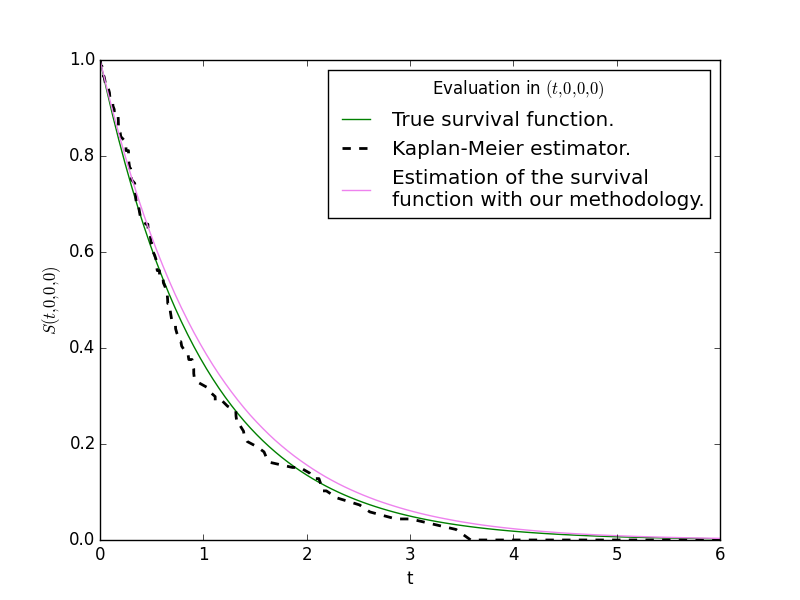} }
	 \vspace{.6ex}
	 \raisebox{-\height}{\includegraphics[width=0.9\textwidth]{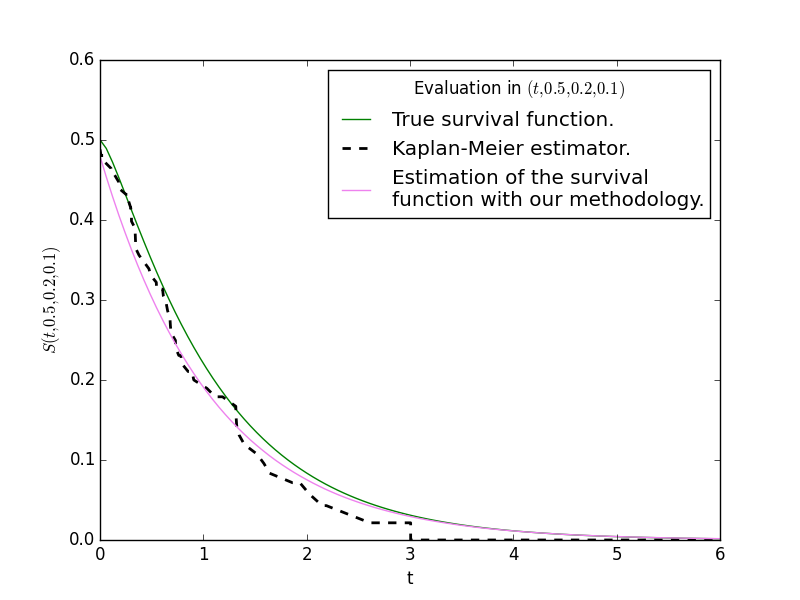} }
	 \vspace{.6ex}
	 \raisebox{-\height}{\includegraphics[width=0.9\textwidth]{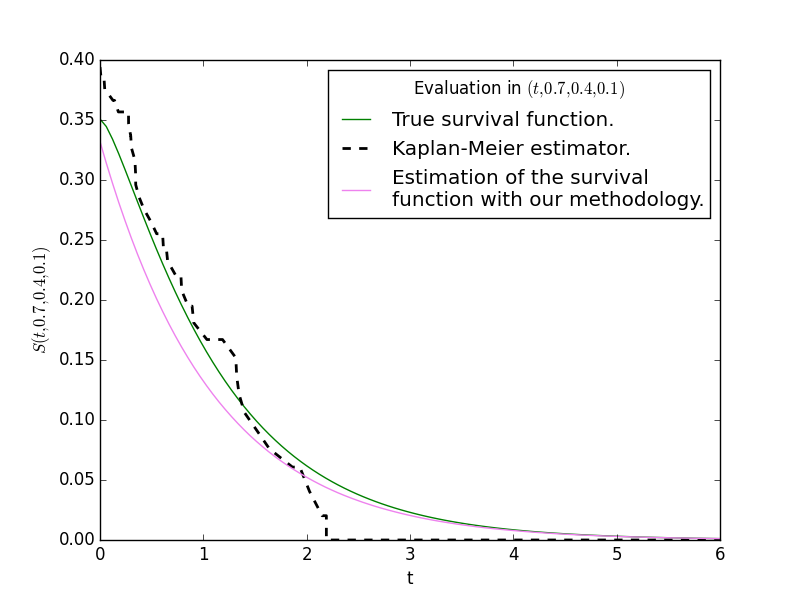} }
    \caption{Fits with first dimension not fixed.}
\end{subfigure}
\hspace{1em}
\begin{subfigure}{0.48\textwidth}
\centering
	\raisebox{-\height}{ \includegraphics[width=0.9\textwidth]{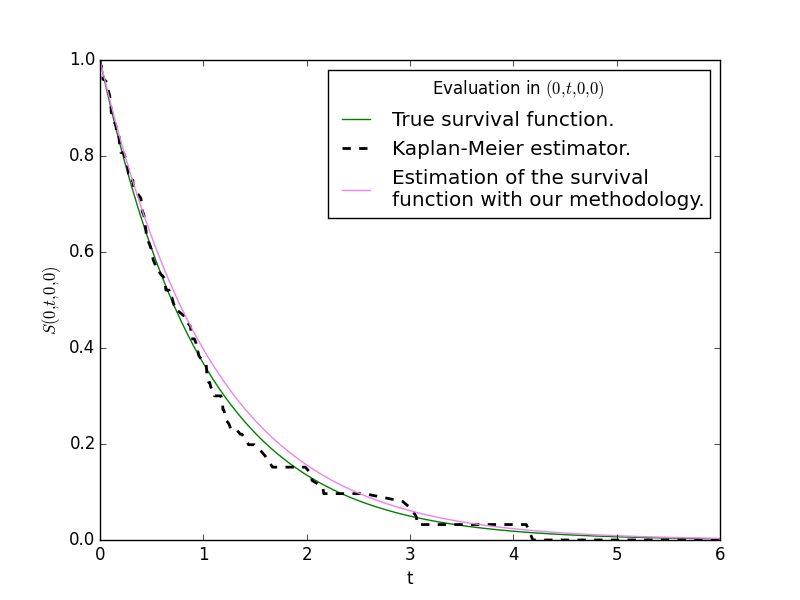} }
	\vspace{.6ex}
	\raisebox{-\height}{ \includegraphics[width=0.9\textwidth]{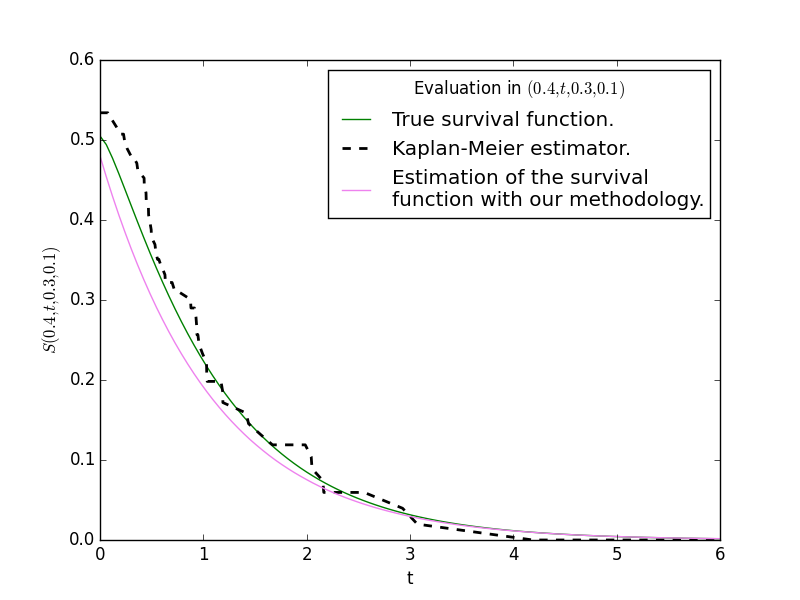} }
	\vspace{.6ex}
	\raisebox{-\height}{ \includegraphics[width=0.9\textwidth]{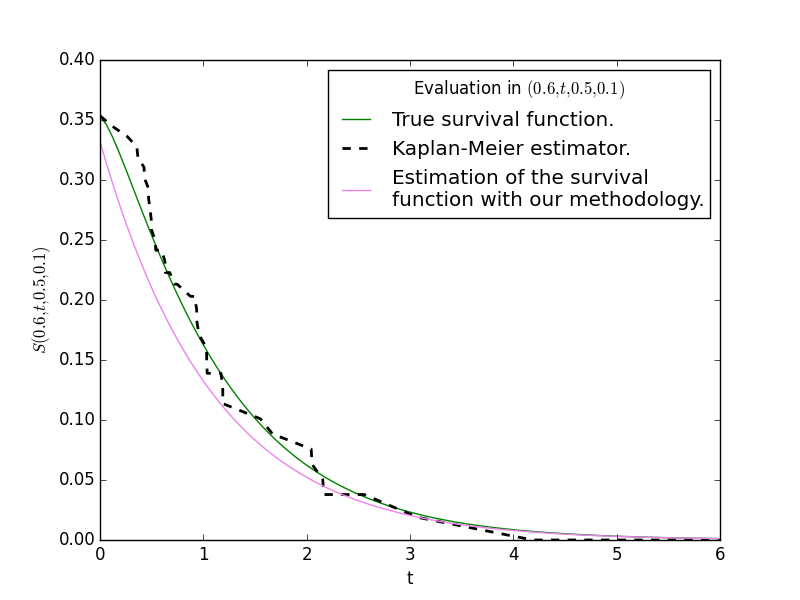} }
    \caption{Fits with second dimension not fixed.}
\end{subfigure}
\caption{Plot of our methodology fits (violet lines), compared with Kaplan-Meier fits (dashed lines) and the true survival function associated to the distributions $F_{\theta=0.3, \lambda = 1.}$ (green lines). The first column shows fits of the survival function with fixed values in all dimensions except the first one; the second column has fixed values in all dimensions except the second one.}
\end{figure}

\begin{figure}
\centering
\begin{subfigure}[]{0.48\textwidth}
\centering
	 \raisebox{-\height}{\includegraphics[width=0.9\textwidth]{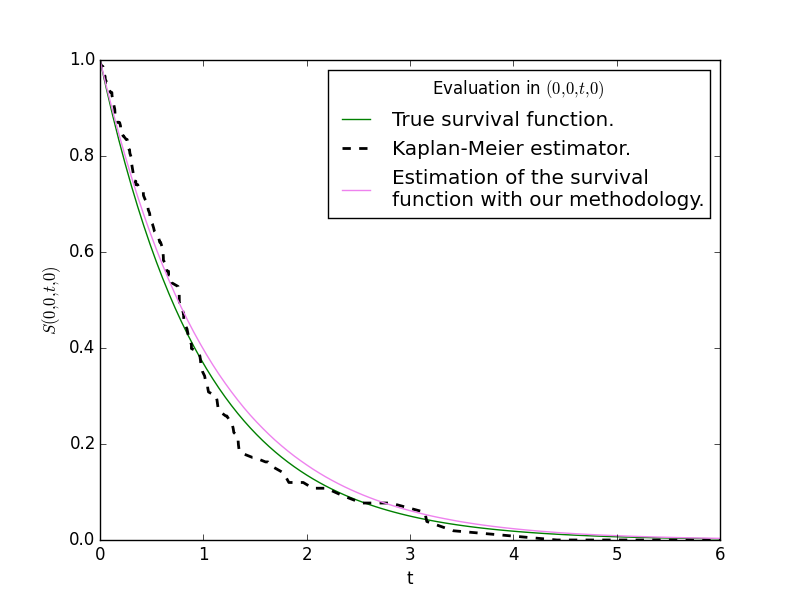} }
	 \vspace{.6ex}
	 \raisebox{-\height}{\includegraphics[width=0.9\textwidth]{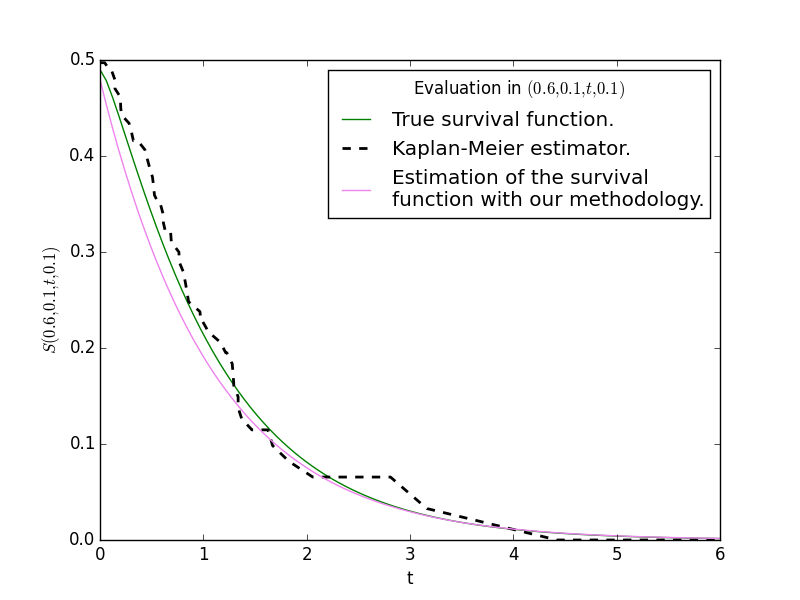} }
	 \vspace{.6ex}
	 \raisebox{-\height}{\includegraphics[width=0.9\textwidth]{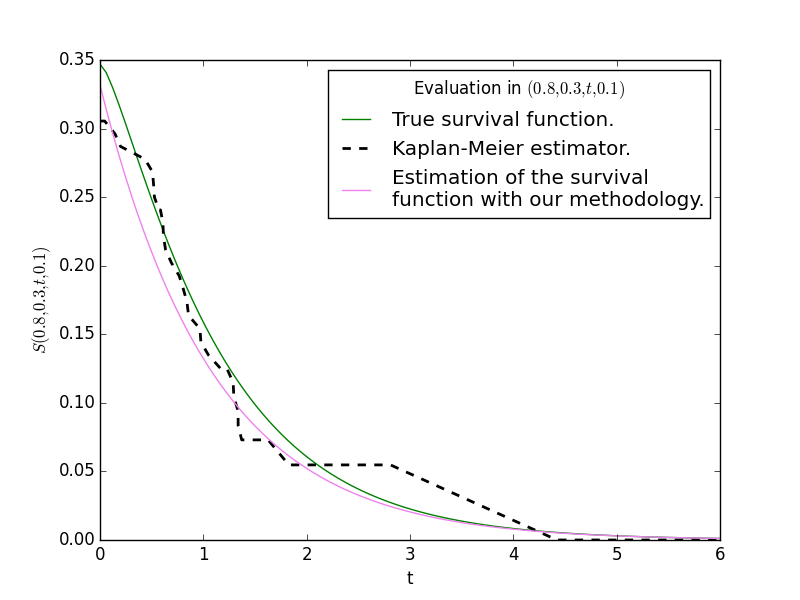} }
    \caption{Fits with third dimension not fixed.}
\end{subfigure}
\hspace{1em}
\begin{subfigure}[]{0.48\textwidth}
\centering
	\raisebox{-\height}{ \includegraphics[width=0.9\textwidth]{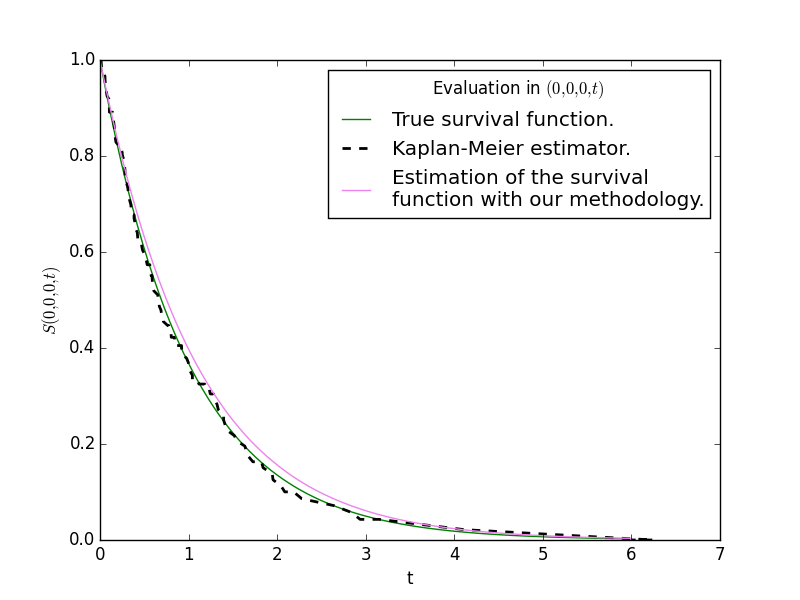} }
	\vspace{.6ex}
	\raisebox{-\height}{ \includegraphics[width=0.9\textwidth]{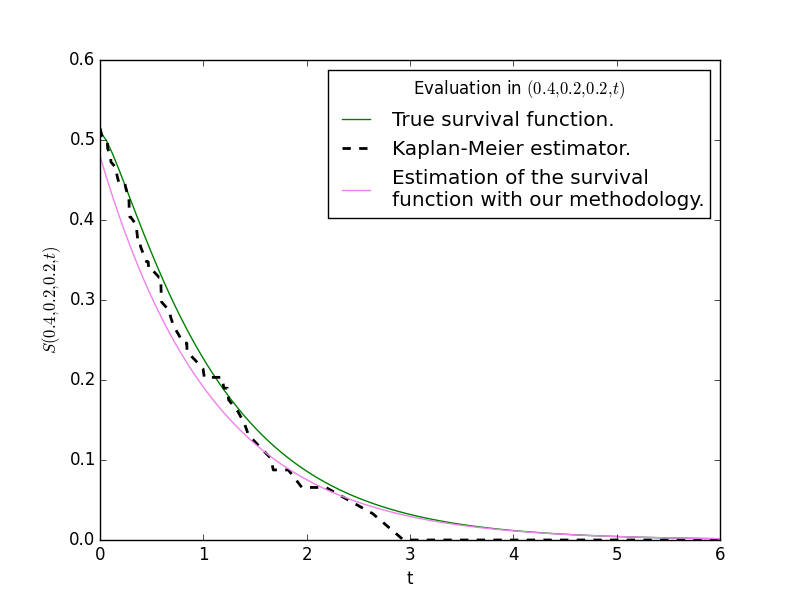} }
	\vspace{.6ex}
	\raisebox{-\height}{ \includegraphics[width=0.9\textwidth]{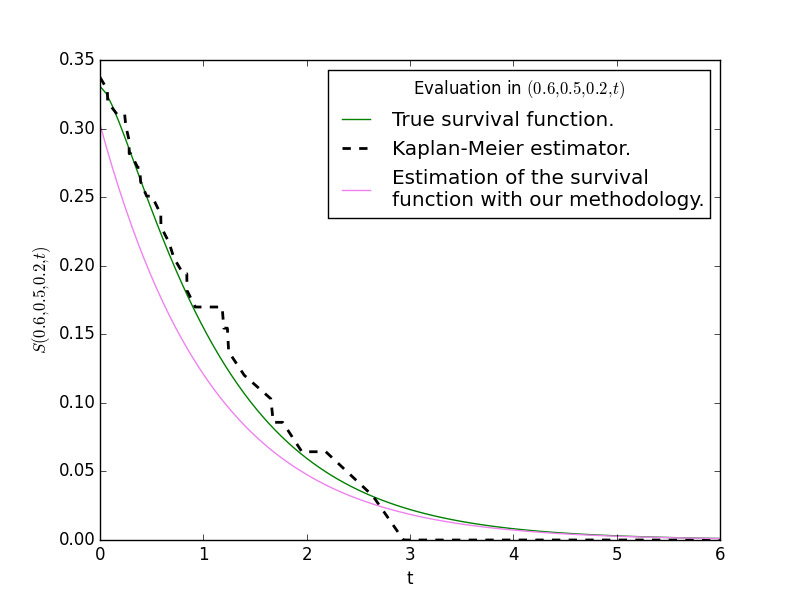} }
    \caption{Fits with fourth dimension not fixed.}
\end{subfigure}
\caption{Plot of our methodology fits (violet lines), compared with Kaplan-Meier fits (dashed lines) and the true survival function associated to the distributions $F_{\theta=0.3, \lambda = 1.}$ (green lines). The first column shows fits of the survival function with fixed values in all dimensions except the third one; the second column has fixed values in all dimensions except the fourth one.}
\end{figure}

\newpage

\section*{Appendix}
\renewcommand{\thesubsection}{A.\arabic{subsection}}
\renewcommand{\theequation}{\Alph{subsection}.\arabic{equation}} 

\subsection{Proof of Proposition 1}\label{proof1}
Given $d\in \{2,3, \dots \}$, we use the notation
$\nu_{-i}(\pmb{s})=\prod_{j=i+1}^d \nu_j(s_j)$
and $\pmb{U}_{k:d}( \pmb{s} ) = \left( U_{k}(s_1), \dots , U_d(s_{d-k+1}) \right)$ for $\pmb{s}\in (\re^+)^d$. Furthermore we define integrals
\begin{equation*}
a_{0,m}(\pmb{\lambda})=\int_{(\re^+)^m}(1-\e^{-<\pmb{\lambda},\pmb{s}>})
\frac{\partial^d}{\partial u_d\cdots\partial u_1}\restr{C_{\theta,m}(\pmb{u})}{\pmb{u}=\pmb{U}_{d-m+1:d}(\pmb{s})}\nu_{-0}(\pmb{s})\mathrm{d}\pmb{s}
\end{equation*}
and
\begin{equation*}
a_{k,m}(\pmb{\lambda})=(-1)^{k+1}\int_{(\re^+)^m}\lambda_1\cdots\lambda_k \e^{-<\pmb{\lambda},\pmb{s}>}
\frac{\partial^{d-k}}{\partial u_d\cdots\partial u_{k+1}}
C_{\theta,m}(\pmb{U}_{d-m+1:d}(\pmb{s}))\nu_{-k}(\pmb{s})\mathrm{d}\pmb{s}
\end{equation*}
where $k\in \{1,\dots, d\}$, $m\in \{0,1,\dots ,d\}$ and $\pmb{\lambda}\in (\re^+)^d$ such that $a_{0,d}(\pmb{\lambda})<\infty$; we also define 
$\prod_{j=k}^la_j=1 \text{ when } k>l$, and denote $\pmb{x}_{-i}$ for the vector $\pmb{x}$ without its $i$-th entry.
\\
An integration by parts shows that
\begin{align*}
a_{0,d}&=-\int_{(\re^+)^{d-1}}\eval{(1-\e^{-<\pmb{\lambda},\pmb{s}>})\frac{\partial^{d-1}}{\partial u_d\cdots\partial u_2}\restr{C_{\theta,d}(\pmb{u})}{\pmb{u}=\pmb{U}_d(\pmb{s})}\nu_{-1}(\pmb{s})}{s_1=0}{s_1=\infty}\mathrm{d}\pmb{s}_{-1}
\\
&\hphantom{=}+\int_{(\re^+)^d}\lambda_1\e^{-<\pmb{\lambda},\pmb{s}>}
\frac{\partial^{d-1}}{\partial u_d\cdots\partial u_2}\restr{C_{\theta,d}(\pmb{u})}{\pmb{u}=\pmb{U}_d(\pmb{s})}\nu_{-1}(\pmb{s})\mathrm{d}\pmb{s}
\\
&=a_{0,d-1}(\pmb{\lambda}_{-1})+a_{1,d}(\pmb{\lambda})
\end{align*}
and in general for $r\in\{1,\dots,d\}$ we get the recursion formula
\begin{equation}\label{recur}
a_{r,d}(\pmb{\lambda})=a_{r,d-1}(\pmb{\lambda}_{-(r+1)})+a_{r+1,d}(\pmb{\lambda})
\numberthis
\end{equation}
We prove the next technical lemma
\begin{lem}
If $a_{0,d}(\pmb{\lambda})<\infty$ then the next $d+1$ identities hold
\begin{align*}
a_{0,d}(\pmb{\lambda})&=\sum_{i=1}^d\psi_i(\lambda_i)-\sum_{\stackrel{\pmb{i}=(i_1,i_2)\in\{1,\dots,d\}^2}{i_1< i_2}}\kappa(\theta;\pmb{\lambda}, \pmb{i}) 
+\cdots
\\
&
\quad \cdots +(-1)^d\sum_{\stackrel{\pmb{i}=(i_1,\dots ,i_{d-1})\in \{1,\dots,d\}^{d-1}}{i_1<\dots <i_{d-1}}} \kappa \left(\theta;\pmb{\lambda}, (i_1,\dots,i_{d-1}) \right)
\\
&
\quad
+(-1)^{d+1}\kappa(\theta;\pmb{\lambda}, (1,\ldots ,d))
\\
a_{1,d}(\pmb{\lambda})&=\psi_1(\lambda_1)-\sum_{i=2}^d\kappa
\left( \theta;\pmb{\lambda}, (1,i) \right)
+\sum_{\stackrel{i_1,i_2\in \{2,\dots , d\}}{i_1<i_2}}\kappa
\left( \theta; \pmb{\lambda}, (1, i_1, i_2) \right) +\cdots 
\\
&\quad \cdots + (-1)^d\sum_{\stackrel{i_1,\dots,i_{d-2}\in \{2,\dots,d\}}{i_1<\cdots<i_{d-2}}}
\kappa(
\left( \theta;\pmb{\lambda}, (1, i_1,\dots ,i_{d-2}) \right)
+(-1)^{d+1}
\kappa
\left( \theta; \pmb{\lambda}, (1,\dots,d) \right) 
\\
&\hphantom{a_{0,s}(\pmb{\lambda})=\int_{(\re^+)^m}(1-\e^{-<\pmb{\lambda},\pmb{x}>})}\vdots
\\
a_{d-1,d}(\pmb{\lambda})&=(-1)^d\kappa
\left( \theta; \pmb{\lambda}, (1,\dots,d-1) \right)
+(-1)^{d+1}\kappa
\left( \theta; \pmb{\lambda}, (1,\dots,d) \right)
\\
a_{d,d}(\pmb{\lambda})&=(-1)^{d+1}\kappa
\left(  \theta; \pmb{\lambda}, (1,\dots, d) \right)
\numberthis  \label{eqsindu}
\end{align*}
\end{lem}
\begin{proof}
We proceed by mathematical induction over the dimension $d$. We observe that from the definition of $\kappa$ we always have
\begin{equation*}
a_{d,d}(\pmb{\lambda})=(-1)^{d+1}\kappa
( \theta; \pmb{\lambda}, (1,\dots,d))
\end{equation*}
For the case $d=2$ we have from Proposition 1 in \cite{epilijoi} that
\begin{equation*}
a_{0,2}(\lambda_1,\lambda_2)=\psi_1(\lambda_1)+\psi_2(\lambda_2)-\kappa(\theta;(\lambda_1,\lambda_2), (1,2))
\end{equation*}
And integrating by parts we obtain
\begin{align*}
a_{1,2}(\lambda_1,\lambda_2)&=\int_{\re^+}\lambda_1\e^{-\lambda_1 s_1}U_1(x_1)\mathrm{d}s_1
-\lambda_1\lambda_2\int_{(\re^+)^2}\e^{-\lambda_1 x_2-\lambda_2 s_2}C_\theta(U_1(s_1),U_2(s_2))\mathrm{d}s_1\mathrm{d}s_2
\\
&=\psi_1(\lambda_1)-\kappa
\left( \theta; \pmb{\lambda}, (1,2) \right)
\end{align*}
Therefore, we get the validity of the equations in (\ref{eqsindu}) for the case $d=2$. Now,  suppose that (\ref{eqsindu}) is true for $d=m-1$, we must show the validity for $d=m$. From the recursion formula (\ref{recur}) we get for $r\in \{0,1,\cdots,d\}$
\begin{align*}
a_{r,m}(\pmb{\lambda})&=a_{r,m-1}(\pmb{\lambda}_{-(r+1)})+a_{r+1,m-1}(\pmb{\lambda}_{-
(r+2)})+\dots +a_{m-1,m-1}(\pmb{\lambda}_{-m})+a_{m,m}(\pmb{\lambda})
\end{align*}
The validity of (\ref{eqsindu}) for $d=m$ follows from the validity for $d=m-1$ and a combinatorial argument.
\end{proof}
Proposition 1 follows by considering the first equation in the Lemma statement and the definition of $a_{0,d}$.
\subsection{Proof of Proposition 2}\label{proof2}
\begin{proof}
Using the independent increments property of CRM's we get that
\begin{align*}
&\prob{Y^{(1)}>t_1,\dots, Y^{(d)}>t_d}=\esp{e^{-\mu_1(0,t_1]-\cdots -\mu_d(0,t_d]}}
\\
& = \esp{\e^{-\mu_{i_1}(0,t_{i_1}]-\cdots -\mu_{i_d}(0,t_{i_1}]}}\esp{\e^{-\mu_{i_2}(t_{i_1},t_{i_2}]-\cdots -\mu_{i_d}(t_{i_1},t_{i_2}]}}\cdots \esp{\e^{-\mu_{i_d}(t_{i_{d-1}},t_{i_d}]}}
\\
&=\e^{-\gamma(t_{i_1})\psi(\pmb{1})}\e^{-[\gamma(t_{i_2})-\gamma(t_{i_1})]\psi_{i_2,\dots, i_d}(\pmb{1})}\cdots e^{-[\gamma(t_{i_d})-\gamma(t_{i_{d-1}})]\psi_{i_d}(\pmb{1})}
\end{align*}
\end{proof}
\subsection{Proof of Proposition 3}\label{proof3}
For notation purposes, in this proof we use the shorthand $\mu(t)=\mu\left( (0,t] \right)$ for a measure $\mu$ and positive real number $t$.
\begin{proof}
For the only if part we define $V_{i,j}=1-\e^{-[\mu_i(t_{i,j})-\mu_i(t_{i,j-1})]}$ for $i\in \{1,\dots,d\}$ and $j\in \{1,\dots ,h\}$ so by supposing $(F_1(t_1),\dots,F_d(t_d))\stackrel{d}{=}(1-\e^{-\mu_{1}(t_1)},\dots,1-\e^{-\mu_{d}(t_d)})$ we have
\begin{align*}
F(t_{1,1},\dots,t_{d,1})&\stackrel{d}{=}[1-\e^{-\mu_1(t_{1,1})}]\cdots[1-\e^{-\mu_d(t_{d,1})}]
\\
&=[1-\e^{-[\mu_1(t_{1,1})-\mu_1(t_{1,0}) ]}]\cdots[1-\e^{-[\mu_d(t_{d,1})-\mu_d(t_{d,0}]}]
\\
&=V_{1,1}\cdots V_{d,1}
\end{align*}
We observe that for $i\in \{2,\dots, h\}$ and $r\in \{1,\dots, d\}$
\begin{align*}
1-\prod_{j=1}^i\bar{V}_{r,j}&=1-\prod_{j=1}^i(1-V_{r,j})=1-\prod_{j=1}^i\e^{-[\mu_r(t_{r,j})-\mu_r(t_{r,j-1}])}=1- \e^{-\mu_r(t_{r,i})}
\end{align*}
So for $i\in \{2,\dots d\}$
\begin{align*}
F(t_{1,i},\dots,t_{d,i})&\stackrel{d}{=}[1-\e^{-\mu_1(t_{1,i})}]\cdots[1-\e^{-\mu_d(t_{d,i})}]
\\&=[1-\prod_{j=1}^i\bar{V}_{1,j}]\cdots [1-\prod_{j=1}^i\bar{V}_{d,j}].
\end{align*}
Concluding the only if part.
\\
For the if part we define $\mu_i(t)=-\log(1-F_i(t))$ for $i\in \{1,\dots , d\}$ and suppose for $h\in\{1,2,\dots\}$, $\pmb{t}_1=(t_{1,1},\dots, t_{d,1}),\dots,\pmb{t}_h=(t_{1,h},\dots,t_{d,h}) $ with 
$t_{0,i}=0<t_{1,i}<\cdots <t_{d,i}$ and 
$t_{j,0}=0<t_{j,1}<\cdots <t_{j,h}$ the existence of independent random vectors  $(V_{1,1},\dots V_{d,1}),\dots , (V_{1,h},\dots V_{d,h})$ such that we have (\ref{ntr}).
\\
Marginalizing in $(\ref{ntr})$, we can apply Theorem 3.1 of \cite{doksum} to each $F_i$ so we obtain that $F_i\sim \text{NTR}(\mu_i)$ for some CRM $\mu_{i}$ that is stochastically continuous, almost surely non-decreasing and has the appropriate limit behaviour.
\\
We observe that 
\begin{align*}
\left(\mu_1(t_j)-\mu_1(t_{j-1}),\dots,\mu_d(t_j)-\mu_d(t_{j-1})
\right)\stackrel{d}{=}\left(-\log(1-V_{1,j}),\dots , -\log(1-V_{d,j})\right)
\end{align*}
Hence $(\mu_{1},\dots,\mu_{d})$ defines a vector of CRM's.
\end{proof}
\subsection{Proof of Theorem 1}\label{proof4}
This proof is not only restricted to the homogeneous L\'evy intensity case; in this general setting, we recall that the Laplace exponent has the form \eqref{laplacenonhomog}. In order to prove the theorem we use the next technical lemma.
\begin{lem}
Let $(\mu_1,\dots,\mu_d)$ be a $d$-variate CRM such that $\mu_1,\dots,\mu_d$ are not independent and let the L\'evy intensity $\nu(\pmb{s},\mathrm{d}t)\mathrm{d}\pmb{s}$ of $(\mu_1,\dots,\mu_d)$ be such that $\eta_t=\nu(\pmb{x},(0,t])$ is differentiable with respect to $t\in \re^+$ at some $t_0\neq 0$ and denote 
$\eta'_{t_0}(\pmb{s})=\partial\restr{\eta_t(\pmb{s}) / \partial t}{t=t_0}$. If $\pmb{q}=(q_1,\dots,q_d)\in \mathbb{N}^d$ are  such that $\max\{q_1,\dots,q_d\}\geq 1$ and $\pmb{r}=(r_1,\dots,r_d)\in ( \re^+ )^d$ are such that $\min\{r_1,\dots,r_d\}\geq 1$, then
\begin{align*}
&\esp{\e^{-r_1\mu_1\left( A_\epsilon \right)-\cdots -r_d\mu_d\left( A_\epsilon \right) }
\left( 1-\e^{-\mu_1\left( A_\epsilon \right)} \right)^{q_1}\cdots
\left( 1-\e^{-\mu_d\left( A_\epsilon \right)} \right)^{q_d}
}
\\
& =\epsilon\int_{(\re^+)^d}\e^{-\langle \pmb{r}, \pmb{s}\rangle }(1-\e^{-s_1})^{q_1}\cdots(1-\e^{-s_d})^{q_d}
\eta'_{t_0}(\pmb{s})\mathrm{d}\pmb{s} +o(\epsilon)
\end{align*}
as $0<\epsilon \to 0$, with $A_\epsilon=(t_0-\epsilon,t_0]$ for some $t_0\in \re^+\setminus \{0\}$.
\end{lem}
\begin{proof}
We denote $\triangle_{s_1}^{s_2} f_t(\pmb{r})=f_{s_2}(\pmb{r})-f_{s_1}(\pmb{r})$ for a function $f$ where $s_1,s_2\in \re^+$ and $\pmb{r}\in \re^d$. We use the binomial theorem and apply expectation to write the left hand side in the equation above as
\begin{align*}
&\sum_{j_1=0}^{q_1}\cdots \sum_{j_d=0}^{q_d}\binom{q_1}{j_1}\cdots\binom{q_d}{j_d}(-1)^{j_1+\cdots+j_d}\e^{-[\psi_{t_0}(r_1+j_1,\dots,r_d+j_d)-\psi_{t_0-\epsilon}(r_1+j_1,\dots,r_d+j_d)]}
\\
&=
\e^{-\triangle^{t_0}_{t_0-\epsilon}\psi_t(\pmb{r})}
+\e^{-\triangle^{t_0}_{t_0-\epsilon}\psi_t(\pmb{r})}
\left\{
\sum_{i=1}^d\sum_{j=1}^{q_i}\binom{q_i}{j}(-1)^{j}\e^{-\triangle^{t_0}_{t_0-\epsilon}[\psi_t(\pmb{r}+j\pmb{e}_i)-\psi_t(\pmb{r})]}
\right.
\\
& \quad
+\sum_{\stackrel{i_1,i_2\in\{1,\dots,d\}}{i_1<i_2}}
\sum_{j_1=1}^{q_{i_1}}\sum_{j_2=1}^{q_{i_2}}
\binom{q_{i_1}}{j_1}\binom{q_{i_2}}{j_2}
(-1)^{j_1+j_2}
\e^{-\triangle^{t_0}_{t_0-\epsilon}[\psi_t(\pmb{r}+j_1\pmb{e}_{i_1}+j_2\pmb{e}_{i_2})-\psi_t(\pmb{r})]}
\\
& \quad
+ \dots
+\left.\sum_{j_1=1}^{q_1}\cdots \sum_{j_d=1}^{q_d}\binom{q_1}{j_1}\cdots\binom{q_d}{j_d}(-1)^{\langle \pmb{1},\pmb{j}\rangle}\e^{-\triangle^{t_0}_{t_0-\epsilon}[\psi_t(\pmb{r}+\pmb{j})-\psi_t(\pmb{r})]}\right\} \numberthis \label{eqprp7}
\end{align*}
We note that for $j_i\in\{0,\dots , x_i\}$, $\,i\in\{1,\dots,d\}$, $\pmb{j}=(j_1,\dots , j_d)$, a Taylor expansion yields
\begin{align*}
&\e^{-\triangle^{t_0}_{t_0-\epsilon}[\psi_t(\pmb{r}+\pmb{j})-\psi_t(\pmb{r})]}
 = \e^{-                                                                                                                                                                                                                                                                                                                                                                                                                                                                                                                                                                                                                                                                                                                                                                                                                                                                                                                                                                                                                                                                                                                                                                                                                                                                                                                                                                                                                                                                                                                                                                                                                                                                                                                                                                                                                                                                                                                                                                                                                                                                                                                                                                                                                                                                                                                                                                                                                                                                                                                                                                                                                                                                                                                                                                                                                                                                                                                                                                                                                                                                                                                                                                                                                                                                                                                                                                                                                                                                                                                                                                                                                                                                                                                                                                                                                                                                                                                                                                                                                                                                                                                                                                                                                                                                                                                                                                                                                                                                                                                                                                                                                                                                                                                                                                                                                                                                                                                                                                                                                                                                                                                                                                                                                                                                                                                                                                                                                                                                                                                                                                                                                                                                                                                                                                                                                                                                                                                                                                                                                                                                                                                                                                                                                                                                                                                                                                                                                                                                                                                                                                                                                                                                                                                                                                                                                                                                                                                                                                                                                                                                                                                                                                                                                                                                                                                                                                                                                                                                                                                                                                                                                                                                                                                                                                                                                                                                                                                                                                                                                                                                                                                                                                                                                                                                                                                                                                                                                                                                                                                                                                                                                                                                                                                                                                                                                                                                                                                                                                                                                                                                                                                                                                                                                                                                                                                                                                                                                                                                                                                                                                                                                                                                                                                                                                                                                                                                                                                                                                                                                                                                                                                                                                                                                                                                                                                                                                                                                                                                                                                                                                                                                                                                                                                                                                                                                                                                                                                                                                                                                                                                                                                                                                                                                                                                                                                                                                                                                                                                                                                                                                                                                                                                                                                                                                                                                                                                                                                                                                                                                                                                                                                                                                                                                                                                                                                                                                                                                                                                                                                                                                                                                                                                                                                                                                                                                                                                                                                                                                                                                                                                                                                                                                                                                                                                                                                                                                                                                                                                                                                                                                                                                                                                                                                                                                                                                                                                                                                                                                                                                                                                                                                                                                                                                                                                                                                                                                                                                                                                                                                                                                                                                                                                                                                                                                                                                                                                                                                                                                                                                                                                                                                                                                                                                                                                                                                                                                                                                                                                                                                                                                                                                                                                                                                                                                                                                                                                                                                                                                                                                                                                                                                                                                                                                                                                                                                                                                                                                                                                                                                                                                                                                                                                                                                                                                                                                                                                                                                                                                                                                                                                                                                                                                                                                                                                                                                                                                                                                                                                                                                                                                                                                                                                                                                                                                                                                                                                                                                                                                                                                                                                                                                                                                                                                                                                                                                                                                                                                                                                                                                                                                                                                                                                                                                                                                                                                                                                                                                                                                                                                                                                                                                                                                                                                                                                                                                                                                                                                                                                                                                                                                                                                                                                                                                                                                                                                                                                                                                                                                                                                                                                                                                                                                                                                                                                                                                                                                                                                                \int_{(\re^+)^d}   \e^{-\langle \pmb{r},\pmb{x}\rangle}
(1-\e^{-\langle \pmb{j},\pmb{s}\rangle })
\triangle^{t_0}_{t_0-\epsilon}\eta_t(\pmb{s})\mathrm{d}\pmb{s}}
\\
&=1-\epsilon\int_{(\re^+)^d}   \e^{-\langle \pmb{r},\pmb{s}\rangle}
(1-\e^{-\langle \pmb{j},\pmb{s}\rangle })\eta'_{t_0}(\pmb{s})\mathrm{d}\pmb{s} +o(\epsilon)
\numberthis \label{taylor}
\end{align*}
Furthermore by the binomial theorem we get the next $d$ identities
\begin{itemize}
\item [(1)] 
$
\displaystyle
\sum_{i=1}^d\sum_{j=1}^{q}\binom{q}{j}(-1)^{j}
(1-\e^{-js})=-\sum_{i=1}^d (1-\e^{-s})^{q} $
\item [(2)] 
$
\displaystyle
\sum_{\stackrel{i_1,i_2\in\{1,\dots,d\}}{i_1<i_2}}
\sum_{j_1=1}^{q_{i_1}}\sum_{j_2=1}^{q_{i_2}}
\binom{q_{i_1}}{j_1}\binom{q_{i_2}}{j_2}
(-1)^{j_1+j_2}(1-\e^{-j_1s_{i_1}-j_2 s_{i_2}})
\\
=\sum_{\stackrel{i_1,i_2\in\{1,\dots,d\}}{i_1<i_2}}
\left\{
(1-\e^{-s_{i_1}})^{q_{i_1}}+(1-\e^{-s_{i_2}})^{q_{i_2}}-(1-\e^{-s_{i_1}})^{q_{i_1}}(1-\e^{-s_{i_2}})^{q_{i_2}}
\right\} $
\item [] \quad \quad\quad \quad\quad \quad\quad 
\quad\quad \quad\quad \quad\quad \quad\quad \quad
\vdots 
\item [(d-1)] 
$
\displaystyle
\sum_{\stackrel{i_1,\dots,i_{d-1}\in\{1,\dots,d\}}{i_1<\cdots <i_{d-1}}}
\sum_{j_1=1}^{q_{i_1}}\cdots \sum_{j_{d-1}=1}^{q_{i_{d-1}}}\binom{q_{i_1}}{j_1}\cdots\binom{q_{i_{d-1}}}{j_{d-1}}(-1)^{j_1+\cdots +j_{d-1}}
(1-\e^{-j_1s_{i_1}-\cdots-j_{d-1}s_{i_{d-1}} })
\\
=\sum_{\stackrel{i_1,\dots,i_{d-1}\in\{1,\dots,d\}}{i_1<\cdots <i_{d-1}}}
 \Bigg \{
(-1)^{d-1}\sum_{j=1}^{d-1}(1-\e^{-s_{i_j}})^{q_{i_j}}
+ 
\\(-1)^{d-2}\sum_{\stackrel{j_1,j_2\in\{i_1,\dots,i_{d-1}\}}{j_1<j_2}}(1-\e^{-s_{j_1}})^{q_{j_1}}(1-\e^{-s_{j_2}})^{q_{j_2}}
+\cdots -(1-\e^{-s_{i_1}})^{q_{i_1}}\cdots (1-\e^{-s_{i_{d-1}}})^{q_{i_{d-1}}}
\Bigg\}
\\
$
\item [(d)]
$
\displaystyle
\sum_{j_1=1}^{q_1}\cdots\sum_{j_d=1}^{q_d}\binom{q_1}{j_1}\cdots\binom{q_d}{j_d}(-1)^{\langle\pmb{1},\pmb{j}\rangle}
(1-\e^{-\langle\pmb{j},\pmb{s} \rangle })
=(-1)^d\sum_{j=1}^d(1-\e^{-s_j})^{q_j} +
\\
(-1)^{d-1}\sum_{\stackrel{j_1,j_2\in\{1,\dots,d\}}{j_1<j_2}}(1-\e^{-s_{j_1}})^{q_{j_1}}(1-\e^{-s_{j_2}})^{q_{j_2}}+\cdots -(1-\e^{-s_{i_1}})^{q_{i_1}}\cdots (1-\e^{-s_{i_d}})^{q_{i_d}}
$
\end{itemize}
So we have that (\ref{eqprp7}) becomes
\begin{align*}
&\e^{-\triangle^{t_0}_{t_0-\epsilon}\psi_t(\pmb{r})}\left\{ 1 +
\sum_{i=1}^d\sum_{j=1}^{q_i}\binom{q_i}{j}(-1)^{j}
-\epsilon \int_{(\re^+)^d} 
\e^{-\langle \pmb{r},\pmb{s} \rangle}
\sum_{i=1}^d\sum_{j=1}^{q_i}\binom{q_i}{j}(-1)^{j}
(1-\e^{-j_1s_1})\eta'_{t_0}(\pmb{s})\mathrm{d}\pmb{s}
\right.
\\
& + \sum_{\stackrel{i_1,i_2\in\{1,\dots,d\}}{i_1<i_2}}
\sum_{j_1=1}^{q_{i_1}}\sum_{j_2=1}^{q_{i_2}}
\binom{q_{i_1}}{j_1}\binom{q_{i_2}}{j_2}
(-1)^{j_1+j_2}
\\
&-\epsilon \int_{(\re^+)^d} \e^{-\langle \pmb{r},\pmb{s}\rangle}
\sum_{\stackrel{i_1,i_2\in\{1,\dots,d\}}{i_1<i_2}}
\sum_{j_1=1}^{q_{i_1}}\sum_{j_2=1}^{q_{i_2}}
\binom{q_{i_1}}{j_1}\binom{q_{i_2}}{j_2}
(-1)^{j_1+j_2}
(1-\e^{-j_1 s_{i_1}-j_2 s_{i_2}})\eta'_{t_0}(\pmb{s})\mathrm{d}\pmb{s}
\\
&
+\dots +\sum_{j_1=1}^{q_1}\cdots \sum_{j_d=1}^{q_d}\binom{q_1}{j_1}\cdots \binom{q_d}{j_d}(-1)^{\langle \pmb{1}, \pmb{j} \rangle }
\\
&\left.-\epsilon \int_{(\re^+)^d} \e^{-\langle \pmb{r},\pmb{s} \rangle}
\sum_{j_1=1}^{q_1}\cdots \sum_{j_d=1}^{q_d}\binom{q_1}{j_1}\cdots \binom{q_d}{j_d}
(-1)^{\langle\pmb{1},\pmb{j}\rangle}
(1-\e^{-\langle \pmb{j},\pmb{s}\rangle })\eta'_{t_0}(\pmb{s})\mathrm{d}\pmb{s}+o(\epsilon) \right\}
\\
&=\e^{-\triangle^{t_0}_{t_0-\epsilon}\psi_t(\pmb{r})}\left\{ \epsilon \int_{(\re^+)^d} \e^{-\langle \pmb{r},\pmb{s} \rangle }(1-\e^{-s_1})^{q_1}\cdots(1-\e^{-s_d})^{q_d}\eta'_{t_0}(\pmb{s})\mathrm{d}\pmb{s}+o(\epsilon) \right\}
\\
&=\left\{1+o(1) \right\}\left\{ \epsilon \int_{(\re^+)^d} \e^{-\langle \pmb{r},\pmb{s} \rangle }(1-\e^{-s_1})^{q_1}\cdots(1-\e^{-s_d})^{q_d}\eta'_{t_0}(\pmb{s})\mathrm{d}\pmb{s}+o(\epsilon) \right\}
\end{align*}
\begin{align*}
&=\left\{ \epsilon \int_{(\re^+)^d} \e^{-\langle \pmb{r},\pmb{s} \rangle }(1-\e^{-s_1})^{q_1}\cdots(1-\e^{-s_d})^{q_d}\eta'_{t_0}(\pmb{s})\mathrm{d}\pmb{s}+o(\epsilon) \right\}
\end{align*}
\end{proof}
\noindent Define 
\begin{equation*}
\Gamma_{\pmb{D},\epsilon}=\bigcap_{i=1}^d\bigcap_{j=1}^k\left\{
((t_1^{(i)},\delta_1^{(i)},\dots , t_{n_1}^{(i)} ,\delta_{n_1}^{(i)})\; : \; m_i^c\left(\{T_{(j)}\}\right)=n^c_{i,j}  \; , \; m_i^e\left( (T_{(j)}-\epsilon , T_{(j)}] \right) = n^e_{i,j}
\right\}
\end{equation*}
so that
\begin{equation*}
\espc{\e^{-\lambda_1\mu_1(0,t]-\cdots \lambda_d\mu_d (0,t]}}{\pmb{D}}=\lim_{\epsilon \to 0}\frac{\esp{\e^{-\lambda_1\mu_1(0,t]-\cdots -\lambda_d\mu_d(0,t]}\mathbbm{1}_{\Gamma_{\pmb{D},\epsilon}}\left(\pmb{D}\right)}}{\prob{\pmb{D}\in \Gamma_{\pmb{D},\epsilon}}}
\end{equation*}
We observe that defining $T_{(0)}=0$, $\; \bar{n}_{i,k+1}^e=0$  for $i\in \{1,\dots , d\}$ and selecting $\epsilon$ sufficiently small such that
$t\not\in (T_{(j)}-\epsilon,T_{(j)})$ for all $j\in\{1,\dots,k\}$
\\
\begin{align*}
&\espc{\e^{-\lambda_1\mu_1(0,t]-\cdots\lambda_d\mu_d(0,t]}\mathbbm{1}_{\Gamma_{\pmb{D},\epsilon}}\left(\pmb{D}\right)}{(\mu_1,\dots,\mu_d)}
\\
&=\prod_{i=1}^d e^{-\lambda_i\mu_i(0,t]}\prod_{j=1}^k\e^{-n_{i,j}^c\mu_i(0,T_{(j)}]- n_{i,j}^e
\mu_i(0,T_{(j)}-\epsilon]}\left(
1-\e^{-\mu_i(T_{(j)}-\epsilon,T_{(j)}]}
\right)^{n_{i,j}^e}
\\
&=\prod_{i=1}^d e^{
-\lambda_i \mathbbm{1}_{(0,t]}(T_{(k)}) \mu_i(T_{(k)},t]}
\prod_{j=1}^k\Big\{
\e^{
-\lambda_i\mathbbm{1}_{ 
(0,t)
}(T_{(j-1)})
\mu_i(T_{(j-1)},\min \{ t, T_{(j)}-\epsilon \}]
-\lambda_i \mathbbm{1}_{
(0,t]} (T_{(j)}) 
\mu_i(T_{(j)}-\epsilon,T_{(j)}]
} 
\\
&\quad \times
\e^{-n_{i,j}^c\sum_{r=1}^j\left(\mu_i(T_{(r)}-\epsilon,T_{(r)}]+\mu_i(T_{(r-1)},T_{(r)}-\epsilon]
\right)
- n_{i,j}^e \sum_{r=1}^j \mu_i(T_{(r-1)},T_{(r)}-\epsilon]
 -n_{i,j}^e \sum_{r=1}^{j-1} \mu_i(T_{(r)}-\epsilon,T_{(r)}]}
\\
&\quad \times \left(
1-\e^{-\mu_i(T_{(j)}-\epsilon,T_{(j)}]}
\right)^{n_{i,j}^e} \Big\}
\\
&=\prod_{i=1}^d 
\left\{\e^{-\lambda_i\mathbbm{1}_{(0,t]}(T_{(k)})\mu_i(T_{(k)},t]- \sum_{j=1}^k
n_{i,j}^c\sum_{r=1}^j\left(\mu_i(T_{(r)}-\epsilon,T_{(r)}]+\mu_i(T_{(r-1)},T_{(r)}-\epsilon]
\right)}\right.
\\
& \quad \times
\e^{-\sum_{j=1}^k n_{i,j}^e \sum_{r=1}^j \mu_i(T_{(r-1)},T_{(r)}-\epsilon]
-\sum_{j=1}^kn_{i,j}^e \sum_{r=1}^{j-1} \mu_i(T_{(r)}-\epsilon,T_{(r)}]}
\\
& \quad \times
\prod_{j=1}^k
\left\{
\e^{
-\lambda_i \mathbbm{1}_{(0,t)}(T_{(j-1)})
\mu_i(T_{(j-1)},\min \{ t, T_{(j)}-\epsilon\} ]
-\lambda_i\mathbbm{1}_{(0,t]}(T_{(j)})
\mu_i(T_{(j)}-\epsilon,T_{(j)}]
} \left(
\left. 1-\e^{-\mu_i(T_{(j)}-\epsilon,T_{(j)}]}
\right)^{n_{i,j}^e} 
\right\}
\right\}
\\
&=\prod_{i=1}^d \left\{
\prod_{j=1}^k \left\{
\e^{-\left[
\lambda_i\mathbbm{1}_{(0,t]}(T_{(j)})+\bar{n}_{i,j}^c+\bar{n}_{i,j+1}^e
\right]\mu_i(T_{(j)}-\epsilon,T_{(j)}]} \left( 1-\e^{-\mu_i(T_{(j)}-\epsilon,T_{(j)}]}
\right)^{n_{i,j}^e} 
\right\} e^{-\lambda_i\mathbbm{1}_{(0,t]}(T_{(k)})\mu_i(T_{(k)},t]}
\right.
\\
&
\left.
\quad \times
\prod_{j=1}^k \left\{
\e^{
-\lambda_i \mathbbm{1}_{(0,t)}(T_{(j-1)})
\mu_i(T_{(j-1)},\min \{ t, T_{(j)}-\epsilon\} ]
-\bar{n}_{i,j}^c\mu_i(T_{(j-1)},T_{(j)}-\epsilon]-\bar{n}_{i,j}^e \mu_i(T_{(j-1)},T_{(j)}-\epsilon]
}
\right\} \right\}
\end{align*}
So defining 
\begin{align*}
I_{1,\epsilon}&=\prod_{j=1}^k\prod_{i=1}^d\left\{\e^{-\left[
\lambda_i\mathbbm{1}_{(0,t]}(T_{(j)})+\bar{n}_{i,j}^c+\bar{n}_{i,j+1}^e
\right]\mu_i(T_{(j)}-\epsilon,T_{(j)}]} \left( 1-\e^{-\mu_i(T_{(j)}-\epsilon,T_{(j)}]}
\right)^{n_{i,j}^e} 
\right\}
\\
I_{2,\epsilon}&=
\prod_{i=1}^d e^{-\lambda_i\mathbbm{1}_{(0,t]}(T_{(k)})\mu_i(T_{(k)},t]}
\prod_{j=1}^k \left\{
\e^{
-\lambda_i \mathbbm{1}_{(0,t)}(T_{(j-1)})
\mu_i(T_{(j-1)},\min \{ t, T_{(j)}-\epsilon\} ]
-
(
\bar{n}_{i,j}^c+ \bar{n}_{i,j}^e
)
\mu_i(T_{(j-1)},T_{(j)}-\epsilon]
}
\right\}
\end{align*}
We get from the independence property of CRM's that
\begin{equation}\label{espnumerador}
\esp{\e^{-\lambda_1\mu_1(0,t]-\cdots -\lambda_d\mu_d(0,t]}\mathbbm{1}_{\Gamma_{\pmb{D},\epsilon}}\left(\pmb{D}\right)}=\esp{I_{1,\epsilon}}\esp{I_{2,\epsilon}}
\end{equation}
We observe that for $r_i=\lambda_i\mathbbm{1}_{(0,t]}(T_{(j)})+\bar{n}_{i,j}^c+\bar{n}_{i,j+1}^e$, $\,i\in \{1,\dots,d\}$ we have that $\min\{r_1,\dots,r_d\}\geq 1$ and for $j\in\{1,\dots ,k\}$ such that $T_{(j)}$ is an exact observation we have that $\max\{n_{1,j},\dots,n_{d,j}\}\geq 1$ so lemma 2 can be applied yielding
\begin{align*}
&\esp{\prod_{i=1}^d\left\{\e^{-\left[
\lambda_i\mathbbm{1}_{(0,t]}(T_{(j)}+\bar{n}_{i,j}^c+\bar{n}_{i,j+1}^e
\right]\mu_i(T_{(j)}-\epsilon,T_{(j)}]} \left( 1-\e^{-\mu_i(T_{(j)}-\epsilon,T_{(j)}]}
\right)^{n_{i,j}^e} \right\}}
\\
&=\epsilon \int_{(\re^+)^d} \prod_{i=1}^d \left\{ \e^{-[\lambda_i\mathbbm{1}_{(0,t]}(T_{(j)})+\bar{n}_{i,j}^c+\bar{n}_{i,j+1}^e]s_i}(1-\e^{-s_i})^{n_{i,j}^e}
\right\}
\eta'_t{T_{(j)}}(\pmb{s}) \mathrm{d}\pmb{s}+o(\epsilon) \numberthis \label{i12}
\end{align*}
On the other hand, for $j\not \in \mathcal{J}=\{j\, : \, T_{(j)} \text{ is an exact observation}\}$ we have $n_{i,j}^e=0$ so by the continuity of $\eta_t(\pmb{s})$ in $t$ we have
\begin{align*}
&\lim_{\epsilon \to 0}\esp{\prod_{i=1}^d\left\{\e^{-\left[
\lambda_i\mathbbm{1}_{(0,t]}(T_{(j)}+\bar{n}_{i,j}^c+\bar{n}_{i,j+1}^e
\right]\mu_i(T_{(j)}-\epsilon,T_{(j)}]} \left( 1-\e^{-\mu_i(T_{(j)}-\epsilon,T_{(j)}]}
\right)^{n_{i,j}^e} \right\}}
\\
&=\lim_{\epsilon \to 0} \esp{\prod_{i=1}^d\left\{\e^{-\left[
\lambda_i\mathbbm{1}_{(0,t]}(T_{(j)}++\bar{n}_{i,j}^c+\bar{n}_{i,j+1}^e
\right]\mu_i(T_{(j)}-\epsilon,T_{(j)}]} \right\}}=1 \numberthis \label{i11}
\end{align*}
From (\ref{i12}), (\ref{i11}) and the independence property of CRM's we obtain
\begin{align*}
&\lim_{\epsilon \to 0}\esp{I_{1,\epsilon}}=\lim_{\epsilon \to 0}\prod_{j\in \mathcal{J}}\left\{
\epsilon \int_{(\re^+)^d} \prod_{i=1}^d \left\{ \e^{-[\lambda_i\mathbbm{1}_{(0,t]}(T_{(j)})+\bar{n}_{i,j}^c+\bar{n}_{i,j+1}^e]s_i}(1-\e^{-s_i})^{n_{i,j}^e}
\right\}
\eta'_{T_{(j)}}(\pmb{s}) \mathrm{d}\pmb{s}+o(\epsilon)\right\}
\end{align*}
Also by continuity and independence, defining $\pmb{\lambda}=(\lambda_1,\dots,\lambda_d)$, we get
\begin{align*}
&\lim_{\epsilon \to 0}\esp{I_{2,\epsilon}}=
\e^{-[\psi_{t}\left(\mathbbm{1}_{(0,t]} (T_{(k)})\pmb{\lambda}\right)-\psi_{T_{(k)}}\left(\mathbbm{1}_{(0,t]}(T_{(k)})\pmb{\lambda}\right)]}\times
\\
&
\times\prod_{j=1}^k\Bigg\{\e^{
-[
\psi_{t\wedge T_{(j)}}
\left(
\mathbbm{1}_{(0,t]}(T_{(j-1)})\pmb{\lambda}+\bar{\pmb{n}}_{j}^c+\bar{\pmb{n}}_{j}^e
\right)
-\psi_{T_{(j-1)}}\left(\mathbbm{1}_{(0,t]}
(T_{(j-1)})\pmb{\lambda}+\bar{\pmb{n}}_{j}^c+\bar{\pmb{n}}_{j}^e
\right)
]
-[
\psi_{ T_{(j)}}
\left( \bar{\pmb{n}}_{j}^c+\bar{\pmb{n}}_{j}^e
\right)
-
\psi_{t\wedge T_{(j)}}
\left( \bar{\pmb{n}}_{j}^c+\bar{\pmb{n}}_{j}^e
\right)
]
}
\Bigg\}
\end{align*}
So by (\ref{espnumerador}), (\ref{i11}) and (\ref{i12}) we get that
\begin{align*}
&\lim_{\epsilon \to 0} \esp{\e^{-\lambda_1\mu_1(0,t]-\dots -\lambda_d\mu_d(0,t]}\mathbbm{1}_{\Gamma_{\pmb{D},\epsilon}}\left(\pmb{D}\right)}
=\e^{-
\triangle_{T_{(k)}}^t 
\psi_t \left(\mathbbm{1}_{(0,t]} (T_{(k)})\pmb{\lambda}\right) 
-\sum_{j=1}^k
\triangle_{T_{(j-1)}}^{t\wedge T_{(j)}} 
\psi_t \left(\mathbbm{1}_{(0,t]} (T_{(j-1)})\pmb{\lambda}
+\bar{\pmb{n}}_{j}^c+\bar{\pmb{n}}_{j}^e
\right) 
}
\\
&
\times\prod_{j\in \mathcal{J}}\lim_{\epsilon \to 0}\left\{
\epsilon \int_{(\re^+)^d} \prod_{i=1}^d \left\{ \e^{-[\lambda_i\mathbbm{1}_{(0,t]}(T_{(j)})+\bar{n}_{i,j}^c+\bar{n}_{i,j+1}^e]s_i}(1-\e^{-s_i})^{n_{i,j}^e}
\right\}\eta'_{T_{(j)}}(\pmb{s}) \mathrm{d}\pmb{s}+o(\epsilon)\right\}\\
&\hspace{10cm}\times
\e^{-
\sum_{j=1}^k
\triangle_{t\wedge T_{(j)}}^{T_{(j)}} 
\psi_t \left(
\bar{\pmb{n}}_{j}^c+\bar{\pmb{n}}_{j}^e
\right) 
}
\end{align*}
And similarly
\begin{align*}
&\lim_{\epsilon \to 0}\prob{\pmb{D}\in \Gamma_{\pmb{D},\epsilon}}=
\e^{
-\sum_{j=1}^k
\triangle_{T_{(j-1)}}^{T_{(j)}} 
\psi_t \left(
\bar{\pmb{n}}_{j}^c+\bar{\pmb{n}}_{j}^e
\right) 
}
\\
&\qquad \times
\prod_{j\in \mathcal{J}}\lim_{\epsilon \to 0}\left\{
\epsilon \int_{(\re^+)^d} \prod_{i=1}^d \left\{ \e^{-(\bar{n}_{i,j}^c+\bar{n}_{i,j+1}^e)s_i}(1-\e^{-s_i})^{n_{i,j}^e}
\right\}\eta'_t{T_{(j)}}(\pmb{s}) \mathrm{d}\pmb{s}+o(\epsilon)\right\}
\end{align*}
We set $T_{(k+1)}=\infty$ so we conclude
\begin{align*}
&\espc{\e^{-\lambda_1\mu_1(0,t]-\cdots -\lambda_d\mu_d (0,t]}}{\pmb{D}}=\lim_{\epsilon \to 0}\frac{\esp{\e^{-\lambda_1\mu_1(0,t]-\cdots -\lambda_d\mu_d(0,t]}\mathbbm{1}_{\Gamma_{\pmb{D},\epsilon}}\left(\pmb{D}\right)}}{\prob{\pmb{D}\in \Gamma_{\pmb{D},\epsilon}}}
\\
&=
\e^{
-\sum_{j=1}^{k+1}
\triangle_{T_{(j-1)}}^{t\wedge T_{(j)}} 
\left[
\psi_t \left(\mathbbm{1}_{(0,t]} (T_{(j-1)})\pmb{\lambda}
+\bar{\pmb{n}}_{j}^c+\bar{\pmb{n}}_{j}^e
\right) 
-\psi_t \left(
\bar{\pmb{n}}_{j}^c+\bar{\pmb{n}}_{j}^e
\right) 
\right]
}
\\
&\quad \times
\prod_{j\in \mathcal{J}}\lim_{\epsilon \to 0} \left\{
\frac{
\epsilon \int_{(\re^+)^d} \prod_{i=1}^d \left\{ \e^{-[\lambda_i\mathbbm{1}_{(0,t]}(T_{(j)})+\bar{n}_{i,j}^c+\bar{n}_{i,j+1}^e]s_i}(1-\e^{-s_i})^{n_{i,j}^e}
\right\}\eta'_{T_{(j)}}(\pmb{s}) \mathrm{d}\pmb{s}+o(\epsilon)
}
{
\epsilon \int_{(\re^+)^d} \prod_{i=1}^d \left\{ \e^{-[\bar{n}_{i,j}^c+\bar{n}_{i,j+1}^e]s_i}(1-\e^{-s_i})^{n_{i,j}^e}
\right\}\eta'_{T_{(j)}}(\pmb{s}) \mathrm{d}\pmb{s}+o(\epsilon)
}
\right\}
\\
&=\e^{-
\sum_{j=1}^{k+1}
\int_{(\re^+)^d \times (T_{(j-1)},t\wedge T_{(j)}] } 
\mathbbm{1}_{(0,t]} (T_{(j-1)})
\left(
1-\e^{-\langle \pmb{\lambda} \, , \, \pmb{s} \rangle}
\right)\e^{-\langle \bar{\pmb{n}}_{j}^c+\bar{\pmb{n}}_{j}^e\, , \, \pmb{s} \rangle }
\nu(\mathrm{d}\pmb{s}, \mathrm{d}u)
}
\\
&\quad \times
\prod_{j\in \mathcal{J}}\left\{
\frac{
 \int_{(\re^+)^d} \prod_{i=1}^d \left\{ \e^{-[\lambda_i\mathbbm{1}_{(0,t]}(T_{(j)})+\bar{n}_{i,j}^c+\bar{n}_{i,j+1}^e]s_i}(1-\e^{-s_i})^{n_{i,j}^e}
 \right\}\eta'_{T_{(j)}}(\pmb{s}) \mathrm{d}\pmb{s}
}
{
 \int_{(\re^+)^d} \prod_{i=1}^d \left\{ \e^{-[\bar{n}_{i,j}^c+\bar{n}_{i,j+1}^e]s_i}(1-\e^{-s_i})^{n_{i,j}^e}\right\}\eta'_{T_{(j)}}(\pmb{s}) \mathrm{d}\pmb{s}
}
\right\}
\end{align*}

\subsection{Simulation Algorithm}\label{algorithmappendix}

We use a Metropolis within Gibbs sampler to draw simulations from $\sigma|\pmb{\mathcal{D}}$ and $A|\pmb{\mathcal{D}}$ as in Section 4. We recall that Corollary 2 gives the likelihood $l(\sigma, A;\pmb{\mathcal{D}})$ and we denote $p_\sigma$, $p_A$ for the prior distributions of $\sigma$ and $A$ as in Section 4. Given initial values 
$\sigma^{(0)}$, $A^{0)}$, the algorithm is as follows
\begin{enumerate}
\item Draw $A^{(i+1)}$ from a Metropolis-Hastings sampler with proposal distribution $g(x'|x)\sim \text{Log-Norm}(\log(x), 1)$ and target distribution $$l(\sigma^{(i)},x;\pmb{\mathcal{D}})p_A(x).$$
\item Draw $\sigma^{(i+1)}$ from a Metropolis-Hastings sampler with Uniform proposal distribution and target distribution
$$
l(x,A^{(i+1)};\pmb{\mathcal{D}} )p_\sigma (x).
$$
\end{enumerate}
For the fits in Section 4 we used 100 iterations for each inner Metropolis-Hasting sampler and 1000 iterations for the overall Gibbs sampler.

\subsection{Survival function of $F_{\theta,\lambda}$.}\label{survappendix}
Let $\mathcal{C}_{\theta,d}$ be a $d$-dimensional distributional Clayton copula and $\tilde{F}_i$, $i=1,\ldots, d$, a collection of marginal cumulative distribution functions; then the survival function associated to the Clayton distributional copula and marginals is given by
$$
S\left( x_1, \ldots , x_d \right) = 1 - \sum_{i=1}^d \tilde{F}_i(x_i)+ \sum_{j=2}^d (-1)^j \sum_{\stackrel{i_1,\ldots, i_j\in\{1,\dots,d\}}{i_1<\ldots < i_j}}
\mathcal{C}_{\theta, j}(x_{i_1},\ldots, x_{i_j}),
$$ 
see Section 2.6 in \cite{nelsen}.
\bibliographystyle{plainnat}

\end{document}